\documentclass[prl,twocolumn,
 notitlepage,groupedaddress,
aps, reprint,
superscriptaddress,
amsmath,amssymb,
floatfix]{revtex4-2}
\usepackage{pifont}%
\usepackage{bbold}%
\usepackage{braket}
\usepackage[dvipsnames]{xcolor} %
\usepackage{amsmath, amssymb, graphicx, bm}
\usepackage{comment}
\usepackage{pifont}%
\usepackage{comment}
\usepackage{amsthm}
\usepackage{hyperref}

\newtheorem{prop}{Proposition}
\newtheorem{theorem}{Theorem}

\begin{document}

    \title{Calculating trace distances of bosonic states in Krylov subspace}
    
    \author{Javier Martínez-Cifuentes}
    \affiliation{Département de génie physique, \'Ecole polytechnique de Montr\'eal, Montr\'eal, QC, H3T 1J4, Canada}
    
    \author{Nicol\'as Quesada}
    \email{nicolas.quesada@polymtl.ca}
    \affiliation{Département de génie physique, \'Ecole polytechnique de Montr\'eal, Montr\'eal, QC, H3T 1J4, Canada}
    
    \date{\today}

    \begin{abstract}

Continuous-variable quantum systems are central to quantum technologies, with Gaussian states playing a key role due to their broad applicability and simple description via first and second moments. Distinguishing Gaussian states requires computing their trace distance, but no analytical formula exists for general states, and numerical evaluation is difficult due to the exponential cost of representing infinite-dimensional operators.
We introduce an efficient numerical method to compute the trace distance between a pure and a mixed Gaussian state, based on a generalized Lanczos algorithm that avoids explicit matrix representations and uses only moment information. The technique extends to non-Gaussian states expressible as linear combinations of Gaussian states. We also show how it can yield lower bounds on the trace distance between mixed Gaussian states, offering a practical tool for state certification and learning in continuous-variable quantum systems.
    \end{abstract}
    
    \maketitle

    \textit{Introduction}--- Continuous variable (CV) systems~\cite{serafini2023quantum, cerf2007quantum}, such as photonic systems, are of central importance for the development of quantum technologies, notably in the fields of quantum communication~\cite{usenko2025continuous, jain2022practical, madsen2012continuous, vaidman1994teleportation}, quantum metrology~\cite{fadel2025quantum, aasi2013enhanced, meyer2001experimental, mccormick2019quantum} and quantum information processing~\cite{braunstein2005quantum, weedbrook2012gaussian, andersen2010continuous, asavanant2022optical}. These systems have also attracted significant theoretical and experimental interest due to their promise of offering a robust and scalable platform for the implementation of fault-tolerant quantum computation~\cite{psiquantum2025manufacturable, bartolucci2023fusion, bourassa2021blueprint, andersen2015hybrid, menicucci2014fault, knill2001scheme, menicucci2006universal, nielsen2004optical, gottesman2001encoding, lloyd1999quantum}, as well as for having the potential to demonstrate near-term quantum computational advantage, particularly at sampling tasks such as Boson Sampling~\cite{aaronson2011computational, hamilton2017gaussian, madsen2022quantum, zhong2020quantum, zhong2021phase,deng2023gaussian, liu2025robust}. 
    
    Among the possible states associated with CV quantum mechanical systems, Gaussian states~\cite{serafini2023quantum, ferraro2005gaussian} have a simple mathematical description, and commonly arise in the modeling of a wide range of phenomena. In the context of photonics, for instance, they describe many of the light sources that are experimentally available (e.g. coherent, thermal or squeezed light)~\cite{barnett2002methods, quesada2022beyond}. These states are also a central building block in the generation of non-Gaussian states that act as resources for several quantum computational tasks~\cite{walschaers2021non, tiedau2019scalability, konno2024logical}.

    Gaussian states are completely parametrized by a \textit{vector of first moments}, $\bm{r}$, and a \textit{covariance matrix}, $\bm{V}$. In principle, no matter the number of modes in the system, the pair $(\bm{r},\bm{V})$ can be experimentally estimated using homodyne or heterodyne measurements~\cite{serafini2023quantum, ferraro2005gaussian, weedbrook2012gaussian}, implying that the characterization of these states is readily available in the laboratory. 
    
    Given two Gaussian states $\hat{\varrho}_1$ and $\hat{\varrho}_2$, it is natural to investigate how well we can distinguish between them when having access to their parametrization in terms of the pairs $(\bm{r}_1, \bm{V}_1)$ and $(\bm{r}_2, \bm{V}_2)$. This is equivalent to asking what is the probability of successfully discriminating between $\hat{\varrho}_1$ and $\hat{\varrho}_2$, and how to express it as a function of the covariance matrices and vectors of first moments. This problem commonly arises, for instance, when comparing an experimentally prepared state with its intended theoretical description. As dictated by the Holevo-Helstrom theorem~\cite{helstrom1969quantum}, the best success probability that we can hope for depends on the \textit{trace distance} between the states. Mathematically, the trace distance between $\hat{\varrho}_1$ and $\hat{\varrho}_2$ is defined as $\frac{1}{2}\|\Delta\hat{\varrho}\|_1$, where $\Delta\hat{\varrho}=\hat{\varrho}_1-\hat{\varrho}_2$, and the trace norm (or 1-norm) of the operator $\hat{A}$ is given by $\|\hat{A}\|_1=\mathrm{Tr}(\sqrt{\hat{A}^\dagger\hat{A}})$. Thus, the task of discriminating between any two Gaussian states boils down to finding a way to compute their trace distance in terms of the data $(\bm{r}_1, \bm{V}_1)$ and $(\bm{r}_2, \bm{V}_2)$.  This calculation also proves relevant in the study of the sample complexity of efficiently estimating properties of these quantum states~\cite{mele2025learning, anshu2024survey}.

    Currently, no closed formula for the trace distance between arbitrary Gaussian states in terms of $(\bm{r}_1, \bm{V}_1)$ and $(\bm{r}_2, \bm{V}_2)$ is known, although a simple equation has been established for the case in which both states are pure~\cite{weedbrook2012gaussian}. On this account, there have been significant efforts to derive analytical bounds on the trace distance as functions of these data~\cite{mele2025learning, bittel2025optimal, mele2025symplectic, bittel2025energy, holevo2024estimates, weedbrook2012gaussian, van2024error, becker2021energy}. From the numerical point of view, the computation of $\frac{1}{2}\|\Delta\hat{\varrho}\|_1$ also proves to be challenging. Current techniques rely on finding a matrix representation of $\Delta\hat{\varrho}$ in a countable basis of the corresponding Hilbert space, usually the Fock basis~\cite{bulmer2026simulating}, and diagonalizing it in order to compute the trace norm. This approach requires the definition of a \textit{cutoff} that allows one to truncate the Hilbert space basis. For a single-mode system, for instance, the cutoff, $c$, is chosen so that $\langle n|\Delta\hat{\varrho}|m\rangle$ is negligible for $n,m > c$. For an $M$-mode system, when choosing the same $c$ for all modes, the number of matrix elements that must be evaluated grows as $\mathcal{O}(c^M)$, making the computation of the trace distance too costly for the majority of cases of interest. Other algorithms employ more efficient ways of truncating the Hilbert space basis in order to obtain an approximation of $\frac{1}{2}\|\Delta\hat{\varrho}\|_1$ up to a fixed precision~\cite{quesada2020exact,qi2022efficient,mele2025achievable}. However, they still rely on finding a number of matrix elements that increases exponentially with $M$.
    
    In this work, we circumvent current problems surrounding the numerical calculation of the trace distance between Gaussian states. First, we propose a method for its efficient computation in the case in which one of the states is pure while the other is mixed. This technique completely avoids the need of finding matrix representations of $\Delta\hat{\varrho}$, involves only relatively simple operations over the pairs $(\bm{r}_1, \bm{V}_1)$ and $(\bm{r}_2, \bm{V}_2)$, and has a time complexity that scales polynomially with the number of modes in the system. Our results are based on a generalization of the well-known Lanczos algorithm~\cite{kreuzer1981lanczos}, a Krylov subspace~\cite{saad2003iterative, caruso2021inverse} technique for numerically computing eigenvalues of Hermitian operators. Second, we show that this method can also be used to compute the trace distance between non-Gaussian states, one of them pure, that can be written as a linear combination of Gaussian states. Third, we argue that while this technique cannot be used to obtain the trace distance between two mixed Gaussian states, it can be used to obtain lower bounds on this value.

    \textit{CV systems formalism}--- An $M$-mode CV quantum system is associated to the Hilbert space $\mathcal{H}=L^2(\mathbb{R}^M)\simeq \bigotimes_{k=1}^M L^2(\mathbb{R})$, i.e. the Hilbert space of square-integrable functions from $\mathbb{R}^M$ to $\mathbb{C}$, over which we can define a set of quadrature operators, $\{(\hat{q}_k, \hat{p}_k)\}_{k=1}^M$, one pair for each mode, that satisfy the canonical commutation relations $[\hat{q}_j, \hat{p}_k] = i \hbar \delta_{j,k}$, $[\hat{q}_j, \hat{q}_k] = [\hat{p}_j, \hat{p}_k] = 0$~\cite{serafini2023quantum}. Here, $[\hat{A}, \hat{B}] = \hat{A}\hat{B} - \hat{B}\hat{A}$ is the commutator between the linear operators $\hat{A}$ and $\hat{B}$, $\delta_{j,k}$ stands for the Kronecker delta, and $\hbar/2$ is a constant (the variance of the vacuum). Define the operator vector $\hat{\bm{r}} = (\hat{q}_1,\dots, \hat{q}_M, \hat{p}_1,\dots,\hat{p}_M)$. Then, for any density operator $\hat{\varrho}$ defined over $\mathcal{H}$, the components of its vector of first moments, $\bm{r}\in \mathbb{R}^{2M}$, are computed as $r_k = \mathrm{Tr}(\hat{r}_k\hat{\varrho})$, while the entries of its covariance matrix, $\bm{V}\in \mathbb{R}^{2M,2M}$, are computed as $V_{j,k} = \frac{1}{2}\mathrm{Tr}(\{\hat{r}_j,\hat{r}_k\}\hat{\varrho}) - r_jr_k$, where $\{\hat{A},\hat{B}\}=\hat{A}\hat{B} + \hat{B}\hat{A}$ stands for the anticommutator between $\hat{A}$ and $\hat{B}$. The covariance matrix of any valid quantum state, i.e., a Hermitian, positive semi-definite $\hat{\varrho}$ with trace one, is symmetric and satisfies the uncertainty relation $\bm{V} + i\frac{\hbar}{2}\bm{\Omega}\succeq 0$~\cite{serafini2023quantum}, which also implies that $\bm{V}$ is positive definite. In this inequality, $\bm{\Omega} = \big(\begin{smallmatrix}0&\mathbb{I}_M\\-\mathbb{I}_M&0\end{smallmatrix}\big)$ is the \textit{symplectic form} and $\mathbb{I}_M\in \mathbb{R}^{M, M}$ is the identity matrix.

    The characteristic function, $\chi(\bm{s}, \hat{\varrho})$, of the operator $\hat{\varrho}$ is defined as $\chi(\bm{s}, \hat{\varrho}) = \mathrm{Tr}[\hat{D}(\bm{s})\hat{\varrho}]$, where $\bm{s}\in \mathbb{R}^{2M}$ and $\hat{D}(\bm{s}) = \exp(-i\bm{s}^\mathrm{T}\bm{\Omega}\hat{\bm{r}}/\hbar)$ is the Weyl or displacement operator. The state $\hat{\varrho}$ is said to be a Gaussian state if its characteristic function can be written as 
    $\chi(\bm{s},\hat{\varrho}) = \exp\left(-\frac{1}{2}\bm{s}^\mathrm{T}\bm{\Omega}^\mathrm{T}\bm{V}\bm{\Omega}\bm{s}-i\bm{r}^\mathrm{T}\bm{\Omega}\bm{s}\right)$. This condition implies that Gaussian states are in one-to-one correspondence with the pairs $(\bm{r}, \bm{V})$. If the Gaussian state is pure, i.e., if it can be written as $\hat{\varrho}=|\psi\rangle\langle\psi|$ with $|\psi\rangle\in \mathcal{H}$, its covariance matrix satisfies the relation $\det[(2/\hbar)\bm{V}] = 1$.

    \textit{The trace distance}--- Given two quantum states $\hat{\varrho}_1$ and $\hat{\varrho}_2$ (defined over $\mathcal{H}$) with equal chance of being generated by a source, the minimum probability of committing and error when discriminating between them by means of a quantum measurement is dictated by the Holevo-Helstrom bound~\cite{helstrom1969quantum}: $\Pr_{\text{error, min}}=\frac{1}{2}-\frac{1}{2}d(\hat{\varrho}_1, \hat{\varrho}_2) $. This theorem provides the trace distance with a strong operational interpretation that other distinguishability measures, such as the fidelity, lack. Hence the great interest on finding efficient ways to compute it. If both states are pure, so that $\hat{\varrho}_k = |\psi_k\rangle\langle\psi_k|$ for $k \in \{1,2\}$, one has that $d(\hat{\varrho}_1, \hat{\varrho}_2) = \sqrt{1-|\langle\psi_1|\psi_2\rangle|^2}$~\cite{weedbrook2012gaussian}. In a more general setting, the computation of the trace distance involves the diagonalization of $\Delta\hat{\varrho}$. This operator is Hermitian and compact, so its eigenvalues, $\{\lambda\}$, are real, and its eigenvectors, $\{|\lambda\rangle\}$, constitute an orthonormal basis of $\mathcal{H}$. Note that the set of eigenvectors is countable on account of $\mathcal{H}$ being separable. Using this basis, one can readily see that $d(\hat{\varrho}_1, \hat{\varrho}_2)=\frac{1}{2}\sum_{\lambda}|\lambda|$. Behind this simple expression lies the fact that computing the set $\{\lambda\}$ is not an easy task. In most cases, one needs to compute an approximate matrix representation of $\Delta\hat{\varrho}$, in another orthonormal basis of $\mathcal{H}$, before the diagonalization, a method that generally requires the truncation of the basis used.
    
    More insight into the computation of the trace distance comes from the factorization $\Delta\hat{\varrho} = \hat{U}\left(\sum_{\lambda}\lambda|\lambda\rangle\langle\lambda|\right)\hat{U}^\dagger$, where $\hat{U}$ is the unitary operator that diagonalizes $\Delta\hat{\varrho}$. After some rearrangement, one obtains that $\Delta\hat{\varrho} = \hat{U}\left(\sum_{\lambda>0}\lambda|\lambda\rangle\langle\lambda|\right)\hat{U}^\dagger-\hat{U}\left(\sum_{\lambda<0}|\lambda||\lambda\rangle\langle\lambda|\right)\hat{U}^\dagger=\hat{P} - \hat{Q}$, where the operators $\hat{P}$, $\hat{Q}$ are positive definite and have orthogonal ranges ($\hat{P}\hat{Q} = \hat{Q}\hat{P} = 0$). This decomposition is known as \textit{positive orthogonal decomposition}, and it is common to Hermitian operators. It is not difficult to see that $d(\hat{\varrho}_1, \hat{\varrho}_2) = \frac{1}{2}\mathrm{Tr}(\hat{P} + \hat{Q})$. Moreover, since $\mathrm{Tr}(\Delta\hat{\varrho}) = \mathrm{Tr}(\hat{P}-\hat{Q}) = 0$, we find that $d(\hat{\varrho}_1, \hat{\varrho}_2)=\mathrm{Tr}(\hat{P})=\mathrm{Tr}(\hat{Q})$. Thus, in order to compute the trace distance, we need only find the positive \textit{or} the negative eigenvalues of $\Delta\hat{\varrho}$. This insight is particularly useful for the case in which one of the states involved, say $\hat{\varrho}_1$, is pure. Indeed, we find the following result.
    \begin{theorem}
        Let $\hat{\varrho}_1$ and $\hat{\varrho}_2$ be quantum states defined over the Hilbert space $\mathcal{H}$. Suppose that $\hat{\varrho}_1$ is rank-1, i.e. it represents a pure state, so that $\hat{\varrho}_1 = |\psi_1\rangle\langle\psi_1|$ with $|\psi_1\rangle\in \mathcal{H}$. Then, the operator $\Delta\hat{\varrho}=\hat{\varrho}_1-\hat{\varrho}_2$ has only one positive eigenvalue.
        \label{theo:only_one_positive_eival}
    \end{theorem}
    
    We present a proof of this theorem in the Supplemental Material. Let $\lambda_+$ be the only positive eigenvalue of $|\psi\rangle\langle\psi|-\hat{\varrho}$, then a corollary to Theorem~\ref{theo:only_one_positive_eival} is that $d(|\psi\rangle\langle\psi|,\hat{\varrho}) = \lambda_+$, i.e., the computation of the trace distance between a pure and a mixed state boils down to computing a single, practically isolated eigenvalue. The corresponding normalized eigenvector $\ket{\lambda_+}$ conviniently gives the two-outcome optimal POVM for distinguishing the two states, $\{\ket{\lambda_+} \bra{\lambda_+} ,\hat{\mathbb{I}} - \ket{\lambda_+} \bra{\lambda_+} \}$, with $\hat{\mathbb{I}}$ the identity operator in $\mathcal{H}$.
    Moreover, this result allows us to interpret the calculation of $d(|\psi\rangle\langle\psi|,\hat{\varrho})$ from a variational point of view, as $-\lambda_+$ can be seen as the ``ground state energy" of the ``Hamiltonian" $\hat{\varrho}-|\psi\rangle\langle\psi|$. This approach is useful for computing analytical lower bounds on the trace distance. In the Supplemental Material we present one of such bounds for the case in which both states are Gaussian.

    \textit{Lanczos algorithm}--- The analytical computation of $d(|\psi\rangle\langle\psi|,\hat{\varrho})$ still proves to be a significant challenge, even if it relies on obtaining only the largest magnitude eigenvalue of a compact, Hermitian operator~\footnote{Since $\lambda_+ = \mathrm{Tr}(\hat{P}) = \mathrm{Tr}(\hat{Q})$, it follows that $\lambda_+$ is equal to the sum of the magnitudes of all negative eigenvalues of $|\psi\rangle\langle\psi|-\hat{\varrho}$, which implies that $\lambda_+$ is greater than or equal to each of these absolute values}. However, this same fact opens the door to the numerical computation of $d(|\psi\rangle\langle\psi|,\hat{\varrho})$ using techniques that are specialized in finding this type of eigenvalues. In particular, \textit{Krylov subspace}~\cite{saad2003iterative, caruso2021inverse} techniques are known to quickly approximate the eigenvalues at the edges of the spectrum. These methods are able to find approximate solutions to inverse linear problems of the form $\hat{A}|b\rangle = |c\rangle$, with $|b\rangle, |c\rangle \in \mathcal{H}$ and $\hat{A}$ a linear operator defined over $\mathcal{H}$, by searching over the \textit{$N$th order Krylov subspace}, $\mathrm{span}\{|c\rangle, \hat{A}|c\rangle, \hat{A}^2|c\rangle,\dots, \hat{A}^{N-1}|c\rangle\} \subset \mathcal{H}$, associated to $\hat{A}$ and $|c\rangle$~\cite{caruso2021inverse}. For the computation of eigenvalues, the \textit{Arnoldi} and \textit{Lanczos} algorithms~\cite{saad2003iterative} are widely used iterative methods that compute an orthonormal basis for the Krylov subspace, and use it to find approximations to the largest eigenvalues of $\hat{A}$, improving their accuracy for increasing $N$. In what follows, we will describe Lanczos algorithm, as presented in Ref.~\cite{kreuzer1981lanczos}, since it is specially tailored for Hermitian operators.

    Let $\hat{A}$ be Hermitian and let $|c\rangle$ be a \textit{trial} normalized vector that can be fully expanded in terms of the eigenvectors of $\hat{A}$, and such that $\hat{A}^\ell|c\rangle$ exists for all $\ell \in \mathbb{Z}_{\geq 0}$. At each step $\ell\geq 0$, Lanczos algorithm computes the vectors
    \begin{equation}
        |\tilde{\phi}_{\ell}\rangle = \hat{A}^{\ell}|c\rangle - \sum_{k=1}^{\ell-1} \langle\phi_k|\hat{A}^{\ell}|c\rangle|\phi_k\rangle
        \label{eq:lanczos_arbitrary}
    \end{equation}
    and $|\phi_{\ell}\rangle = |\tilde{\phi}_{\ell}\rangle/\sqrt{\langle\tilde{\phi}_{\ell}|\tilde{\phi}_{\ell}\rangle}$, where $|\phi_0\rangle = |c\rangle$. Afterwards, the matrix $\bm{T}_{\ell}\in \mathbb{R}^{\ell+1, \ell+1}$, with entries $T_{j,k}^{(\ell)} = \langle\phi_j|\hat{A}|\phi_k\rangle$ for $j,k \in \{0, \dots, \ell\}$, is constructed and diagonalized. This matrix is real and \textit{tridiagonal} on account of $\hat{A}$ being Hermitian. The resulting eigenvalues, $\{t^{(\ell)}\}$, are approximations to the eigenvalues of $\hat{A}$, with guaranteed convergence at the limit $\ell\rightarrow\infty$~\cite{kreuzer1981lanczos}. 

    For our study, the eigenvectors of $\hat{A} = |\psi\rangle\langle\psi|-\hat{\varrho}$ constitute an orthonormal basis of $\mathcal{H}$, so any trial vector can be fully expanded in terms of them. This gives us a significant freedom in the choice of the trial vector, while still guaranteeing that $\hat{A}^\ell|c\rangle$ is well defined for any $\ell$. For simplicity, let us set $|c\rangle=|\psi\rangle$. With these choices, we find the following expressions for $\hat{A}^\ell|c\rangle$ and $|\tilde{\phi}_\ell\rangle$:
    \begin{equation}
        \hat{A}^\ell|c\rangle = \sum_{k=0}^\ell C_{\ell, k}\hat{\varrho}^k|\psi\rangle, \quad |\tilde{\phi}_\ell\rangle = \sum_{k=0}^\ell \tilde{D}_{\ell, k}\hat{\varrho}^k|\psi\rangle,
        \label{eq:lanczos_vectors_update}
    \end{equation}
    where the entries of the matrices $\bm{C}, \tilde{\bm{D}}\in \mathbb{R}^{\ell+1, \ell+1}$ are recursively computed using the relations
    \begin{equation}
        C_{\ell,0} = (\bm{C}\bm{\mathcal{G}})_{\ell-1, 0},\;C_{l,k} = -C_{\ell-1, k-1}\text{ for }1\leq k\leq \ell,
        \label{eq:recurrence_matrix_C}
    \end{equation}
    \begin{equation}
        \tilde{D}_{\ell,\ell} = C_{\ell, \ell},\; \tilde{D}_{l,k} = (\bm{C}-\bm{C}\bm{\mathcal{G}}\tilde{\bm{D}}^\mathrm{T}\tilde{\bm{D}})_{\ell, k}\text{ for }0\leq k\leq \ell-1
        \label{eq:recurrence_matrix_D}
    \end{equation}
    with the initial conditions $C_{0,0}=\tilde{D}_{0,0}=1$. Notice that with each iteration, new entries are added to $\bm{C}$ and $\tilde{\bm{D}}$. In order to keep their shape $(\ell+1, \ell+1)$ at step $\ell$, we must set $C_{j,k}=\tilde{D}_{j,k}=0$ for $0\leq j\leq \ell-1$ and  $j+1\leq k\leq \ell$. This means that $\bm{C}$ and $\tilde{\bm{D}}$ are always lower triangular. A proof of Eqs.~\eqref{eq:lanczos_vectors_update} to ~\eqref{eq:recurrence_matrix_D} can be found in the Supplemental Material.

    We refer to the matrix $\bm{\mathcal{G}}\in \mathbb{R}^{\ell+1, \ell+1}$, whose entries are defined as $(\bm{\mathcal{G}})_{j,k}=\langle\psi|\hat{\varrho}^{j+k}|\psi\rangle$ for $j,k\in \{0,\dots,\ell\}$, as the \textit{metric} associated to $\hat{\varrho}$ and $|\psi\rangle$. The reason behind this nomenclature is that $\bm{\mathcal{G}}$ mediates the inner product between the vectors $|\tilde{\phi}_\ell\rangle$, while storing all the information about the nature of the states involved in the computation of the trace distance. Indeed, we may identify the vector $|\tilde{\phi}_\ell\rangle$ with the $l$th row of the matrix $\tilde{\bm{D}}$, $(\tilde{\bm{d}}_\ell)^{\mathrm{T}}$. Then, we find that $\langle\tilde{\phi}_k|\tilde{\phi}_\ell\rangle = (\tilde{\bm{d}}_k)^\mathrm{T}\bm{\mathcal{G}}\tilde{\bm{d}}_{\ell} = (\tilde{\bm{D}}\bm{\mathcal{G}}\tilde{\bm{D}}^\mathrm{T})_{k,\ell}$. In addition to this, it can be shown that the entries of the matrix $\bm{T}_\ell$ can also be computed in a similar way:
    \begin{align}
        T^{(\ell)}_{j,k} &= \langle\phi_j|(|\psi\rangle\langle\psi|-\hat{\varrho})|\phi_k\rangle\nonumber\\
        &=(\bm{D}\bm{\mathcal{G}})_{j,0}(\bm{D}\bm{\mathcal{G}})_{k,0} - (\bm{D}\bm{\mathcal{G}}'\bm{D}^\mathrm{T})_{j,k}.
        \label{eq:tridiagonal_entries}
    \end{align}
    Here, $\bm{D}$ is obtained from $\tilde{\bm{D}}$ after normalization of the vectors $|\tilde{\phi}_\ell\rangle$, and $\bm{\mathcal{G}}'$ is a \textit{shifted} metric, whose entries read $(\bm{\mathcal{G}}')_{j,k}=\langle\psi|\hat{\varrho}^{j+k+1}|\psi\rangle$.

    \begin{figure}[!t]
        {
        \includegraphics[width=\columnwidth]{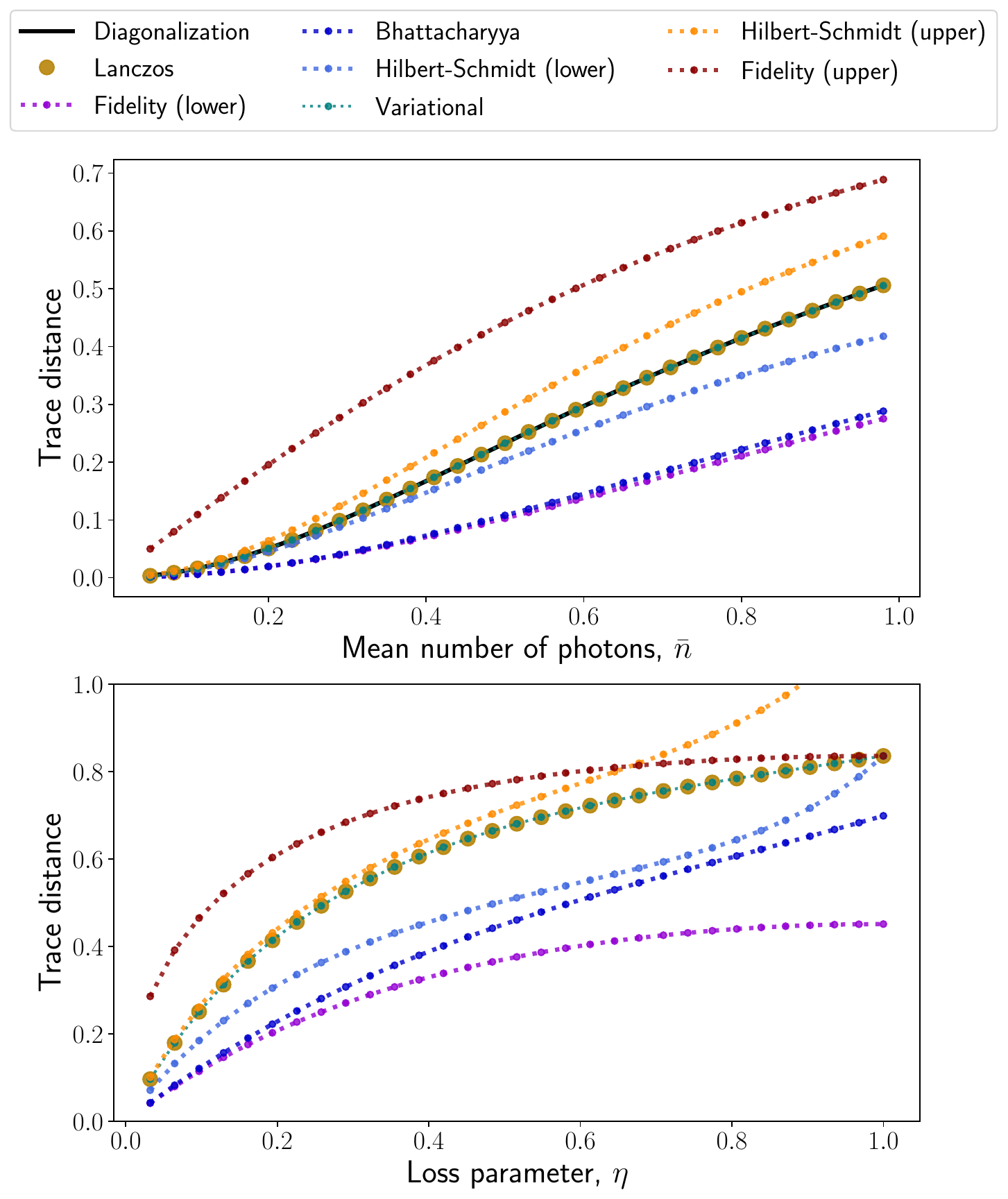}%
        }
        \caption{\textit{(Top)} Trace distance between a single-mode \textit{squashed state}~\cite{martinez2023classical}, $\bm{r} = \bm{0}$ and $\bm{V} = \frac{\hbar}{2}\big(\begin{smallmatrix}1&0\\0&1 + 4\bar{n}\end{smallmatrix}\big)$, and the vacuum state, $\bm{r} = \bm{0}$ and $\bm{V} = \frac{\hbar}{2}\mathbb{I}_2$, as a function of the mean number of photons, $\bar{n}$. \textit{(Bottom)} Trace distance between the $(M=10)$-mode pure Gaussian state parametrized by $\bm{r}_1=\bm{0}$ and $\bm{V}_1=\frac{\hbar}{2}(e^{-2s}\mathbb{I}_M)\oplus(e^{2s}\mathbb{I}_M)$, and the mixed state represented by $\bm{r}_2=\bm{0}$ and $\bm{V}_2=\frac{\hbar}{2}([(1-\eta)e^{-2s}+\eta]\mathbb{I}_M)\oplus([(1-\eta)e^{2s}+\eta]\mathbb{I}_M)$. Results are presented as a function of the loss parameter $\eta\in[0.0,1.0]$, for $s=0.5$. In both figures, the estimations running Lanczos algorithm for $\ell = 10$ steps are shown in yellow circles. In the single-mode case, results using the diagonalization of $|\psi\rangle\langle\psi|-\hat{\varrho}$, in Fock basis with cutoff $c=100$, are represented by a solid black line. Several bounds on the trace distance~\cite{weedbrook2012gaussian}, including the variational bound derived in the Supplemental Material, are shown in dotted lines.} 
        \label{fig:td_gaussian_example}
    \end{figure}

    \textit{Computing the metric}--- Up to this point, we have not put any restrictions over the states $|\psi\rangle\langle\psi|$ and $\hat{\varrho}$ other than requiring them to be valid quantum states. This means that, in principle, our description of Lanczos algorithm remains sound for any pair of states defined over $\mathcal{H}$. However, the actual implementation of the algorithm will rely on how easy it is to obtain $\bm{\mathcal{G}}$ and $\bm{\mathcal{G}}'$ or, equivalently, how easy it is to compute $\langle\psi|\hat{\varrho}^\ell|\psi\rangle$. In what follows, we will show how to carry out this calculation for Gaussian states. 

    Let $\{\hat{\varrho}_k\}_{k=1}^m$ be a set of ordered $M$-mode Gaussian states, all of them defined over $\mathcal{H}$. This set of states is completely characterized by the set of ordered pairs $\{(\bm{r}_k, \bm{V}_k)\}_{k=1}^m$. The \textit{Bargmann invariant}, also known as \textit{multivariate trace}, of $\{\hat{\varrho}_k\}_{k=1}^m$ is defined as $\mathrm{Tr}(\hat{\varrho}_1\hat{\varrho}_2\cdots\hat{\varrho}_m)$. Note that we can immediately link the computation of Bargmann invariants to the computation of $\bm{\mathcal{G}}$ and $\bm{\mathcal{G}}'$ by means of the relation $\langle\psi|\hat{\varrho}^\ell|\psi\rangle = \mathrm{Tr}(\hat{\varrho}^\ell|\psi\rangle\langle\psi|)$, i.e., $\langle\psi|\hat{\varrho}^\ell|\psi\rangle$ is the Bargmann invariant of the set $\{\hat{\varrho},\dots, \hat{\varrho}, |\psi\rangle\langle\psi|\}$, where $\hat{\varrho}$ is repeated $\ell$ times. In terms of $\{(\bm{r}_k, \bm{V}_k)\}_{k=1}^m$, $\mathrm{Tr}(\hat{\varrho}_1\hat{\varrho}_2\cdots\hat{\varrho}_m)$ can be written as~\cite{xu2025bargmann}
    \begin{equation}
        \mathrm{Tr}(\hat{\varrho}_1\hat{\varrho}_2\cdots\hat{\varrho}_m)=\frac{\exp\left(-\tfrac{1}{2}\bm{z}^{\mathrm{T}}\bm{M}^{-1}\bm{z}\right)}{\sqrt{\det(\bm{M}/\hbar)}} ,\label{eq:bargmann_invariant}
    \end{equation}
    where $\bm{z}= \bigoplus_{k=1}^{m-1}(\bm{r}_k - \bm{r}_m)$, and 
    \begin{align}
        \bm{M} = &\bigoplus_{k=1}^{m-1}(\bm{V}_k + \bm{V}_m) +\left[\bm{J}\otimes\left(\bm{V}_m + i\tfrac{\hbar}{2}\bm{\Omega}\right)+ \text{t.}\right].
        \label{eq:bargmann_invariant_matrix}
    \end{align}
    Here, $\text{t.}$ indicates the transpose of the previous terms inside the square brackets, $\bm{A}\otimes\bm{B}$ is the Kronecker product between $\bm{A}$ and $\bm{B}$, $\bm{\Omega}$ is the symplectic form, and the entries of $\bm{J}\in \mathbb{R}^{m-1, m-1}$ are defined as $J_{j,k}=1$ for $k>j$ and $J_{j,k}=0$ otherwise. Note that $\bm{z}\in \mathbb{R}^{2(m-1)M}$ and $\bm{M}\in \mathbb{R}^{2(m-1)M, 2(m-1)M}$, implying that the time complexity of computing Bargmann invariants scales polynomially, typically as $\mathcal{O}(m^3M^3)$. If $\bm{V}_k=\bm{V}$ for all $k$, and $\bm{V}$ corresponds to a pure state, $\det(\bm{M}/\hbar)=1$~\cite{xu2025bargmann}, which further speeds up the calculation. We can see that the time of computation of the metric scales polynomially with the number of modes and with the number of steps we want the algorithm to run for. Moreover, this calculation only requires knowledge of the vectors of first moments and covariance matrices of $|\psi\rangle\langle\psi|$ and $\hat{\varrho}$. 
    
    In Fig.~\ref{fig:td_gaussian_example}, we show some examples of the numerical computation of the trace distance using a Julia implementation of Lanczos algorithm~\footnote{This implementation is available at \url{https://github.com/polyquantique/trace_distance_lczs}}. In the single-mode case, we compare our results with those obtained by direct diagonalization of a matrix representation of $|\psi\rangle\langle\psi|-\hat{\varrho}$ in the Fock basis with cutoff $c=100$. We notice a great agreement between the two techniques, even when running Lanczos algorithm for a few number of steps ($\ell = 10$). For the multimode case, we make a comparison with several of the bounds reported in Ref.~\cite{weedbrook2012gaussian}.

    \begin{figure}[!t]
        {
        \includegraphics[width=\columnwidth]{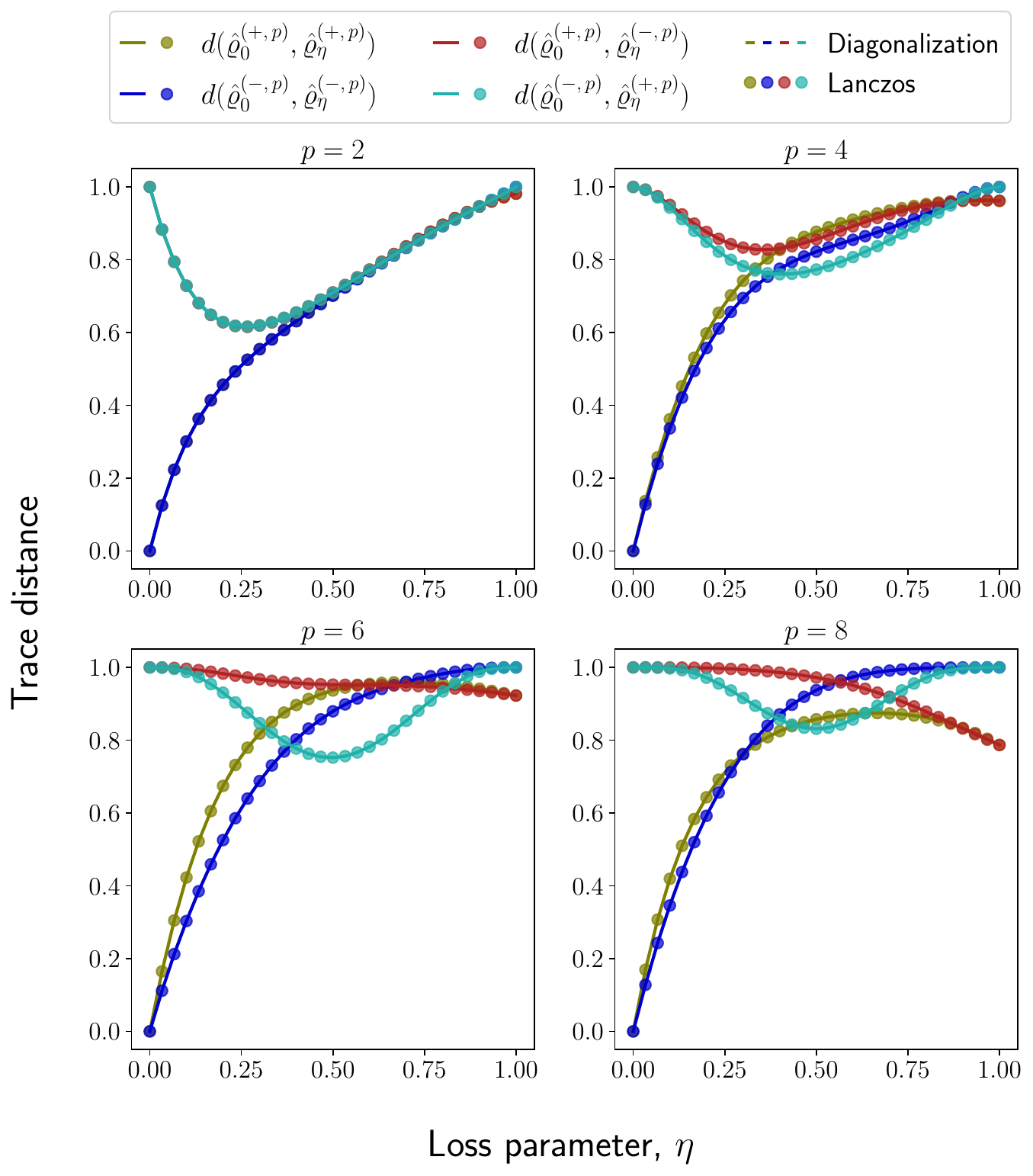}%
        }
        \caption{Trace distance between the single-mode \textit{cat states} $|\mathcal{C}_{\pm}^{(p)}\rangle = (\sqrt{N_{\pm}})^{-1/2}\sum_{j=0}^{p-1}(\pm 1)^{j}|\alpha e^{i2\pi j/p}\rangle$, and the mixed states $\hat{\varrho}_\eta^{(\pm,\,p)}=\sum_{j,k=0}^{p-1}b_{j,k}^{\pm}|\sqrt{1-\eta}\alpha e^{i2\pi j/p}\rangle\langle \sqrt{1-\eta}\alpha e^{i2\pi k/p}|$, with $b_{j,k}^{\pm} = (\pm1)^{j+k}e^{-\eta|\alpha|^2}\exp(\eta|\alpha|^2e^{i2\pi (j-k)/p})/N_{\pm}$, which are obtained after sending the cat states through a loss channel, $\mathcal{L}_\eta[\cdot]$, with loss parameter $\eta$, i.e., $\hat{\varrho}_\eta^{(\pm,\,p)}=\mathcal{L}_{\eta}[|\mathcal{C}_{\pm}^{(p)}\rangle\langle\mathcal{C}_{\pm}^{(p)}|]$. 
        Here, $|\alpha\rangle$ is a coherent state, $\bm{r}=\sqrt{2\hbar}(\mathrm{Re}(\alpha), \mathrm{Im}(\alpha))$ and $\bm{V}=\frac{\hbar}{2}\mathbb{I}_2$, and $N_{\pm}$ are normalization constants. The estimated results are shown as a function of $\eta\in[0.0, 1.0]$ for $\alpha = 2.0$, and $p\in\{2, 4, 6, 8\}$. Solid lines correspond to the diagonalization of $\hat{\varrho}_0^{(\pm,\,p)}-\hat{\varrho}_\eta^{(\pm,\,p)}$ in Fock basis with cutoff $c=100$. Circles correspond to an estimation running Lanczos algorithm for $\ell = 10$ steps.}
        \label{fig:td_cats_example}
    \end{figure}

    The efficient computation of $\bm{\mathcal{G}}$ and $\bm{\mathcal{G}}'$ can also be done for the situation in which $|\psi\rangle\langle\psi|$ and $\hat{\varrho}$ can be written as \textit{linear combinations} of outer products of pure Gaussian states. This type of states are known to efficiently simulate non-Gaussian states of interest to the field of quantum information processing, such as Fock states, multi-component cat states, and GKP states~\cite{bourassa2021fast, solodovnikova2025fast, albert2018performance}. More precisely, suppose that $|\psi\rangle = \sum_{j=1}^p a_j|g_j\rangle$ and $\hat{\varrho}=\sum_{j,k = 1}^qb_{j,k}|f_j\rangle\langle f_k|$, with $a_j, b_{j,k}\in \mathbb{C}$ and $|g_j\rangle$, $|f_j\rangle$ pure Gaussian states. We assume that the $\{a_j\}$, $\{b_{j,k}\}$ are chosen so that $|\psi\rangle$, $\hat{\varrho}$ are normalized and, without loss of generality, that the vacuum amplitudes of the Gaussian states satisfy $\langle0|g_j\rangle\geq 0$, $\langle0|f_j\rangle\geq 0$ for all $j$~\cite{dias2024classical,quesada2025s,yao2024riemannian, sun2026representation}. 

    For these choices of $|\psi\rangle$ and $\hat{\varrho}$, it holds that
    \begin{align}
        \langle\psi|\hat{\varrho}^\ell|\psi\rangle = \sum_{j,k=1}^p\sum_{m,n=1}^qa_j^*a_kb_{m,n}^{(\ell)}\langle g_j|f_m\rangle\langle f_n|g_k\rangle,
        \label{eq:linear_comb_metric}
    \end{align}
    where the coefficients $b_{m,n}^{(\ell)}$ can be recursively computed as
    \begin{align}
        b_{m,n}^{(\ell)} = \sum_{j,k=1}^qb_{m,j}b_{k,n}^{(\ell-1)}\langle f_j|f_k\rangle,
        \label{eq:recursive_power_expectation_coeffs}
    \end{align}
    and the inner products $\langle f_j|g_k\rangle$, $\langle f_j|f_k\rangle$ can be efficiently obtained using the \textit{Bargmann representation} of Gaussian states~\cite{yao2024riemannian}. These expressions show that the time of computation of $\bm{\mathcal{G}}$ and $\bm{\mathcal{G}}'$ grows polynomially with $p$, $q$, the number of modes in the system, and the number of steps Lanczos algorithm runs for. This time complexity can be explained by noticing that for any $\ell$, the range of $\hat{\varrho}^\ell$ is $\mathrm{span}(\{|f_j\rangle\}_{j=1}^q)$, so there is no need of exploring a different subspace of $\mathcal{H}$ at every iteration of the algorithm.

    In Fig.~\ref{fig:td_cats_example}, we show an example of the computation of the trace distance between linear combinations of outer products of Gaussian states. Namely, we estimate the trace distance between pure and lossy \textit{multi-component} cat states~\cite{bergmann2016quantum, boon2026generalised}. We compare our results, running the algorithm for $\ell = 10$ steps, with those obtained by direct diagonalization of $|\psi\rangle\langle\psi|-\hat{\varrho}$ in Fock basis with cutoff $c=100$, and observe a great agreement between the two methods.  

    One may also extend our techniques to the case in which the mixed state $\hat{\varrho}$ can be written as $\hat{\varrho}=\sum_{j=1}^q b_j\hat{\nu}_j$, where $b_j\in \mathbb{C}$ and $\hat{\nu}_j$ is a \textit{product} of mixed Gaussian states. In this case the time of computation of the metric grows exponentially with the number of steps that the algorithm runs for (see Supplemental Material for details).
    
    \begin{figure}[!t]
        {
        \includegraphics[width=\columnwidth]{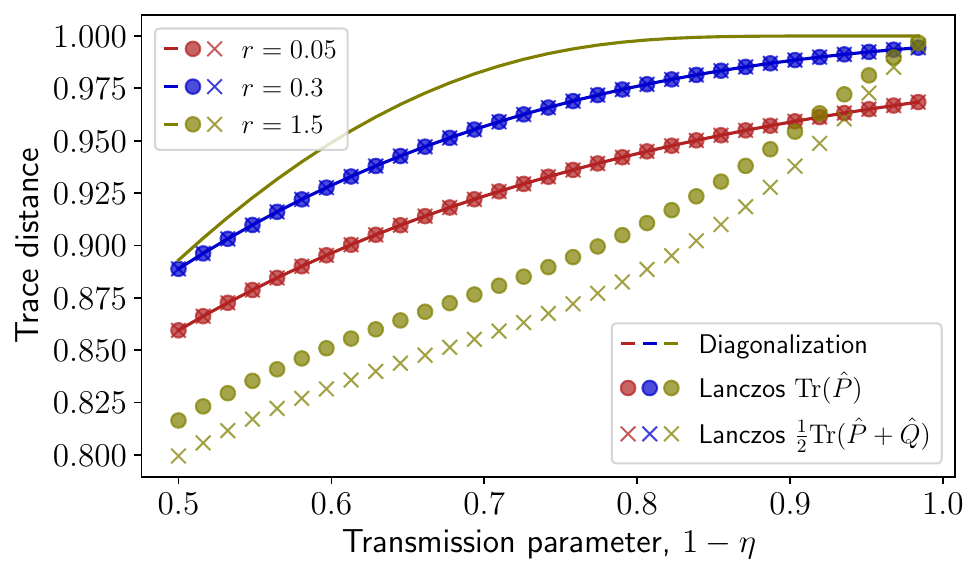}%
        \\
        \includegraphics[width=\columnwidth]{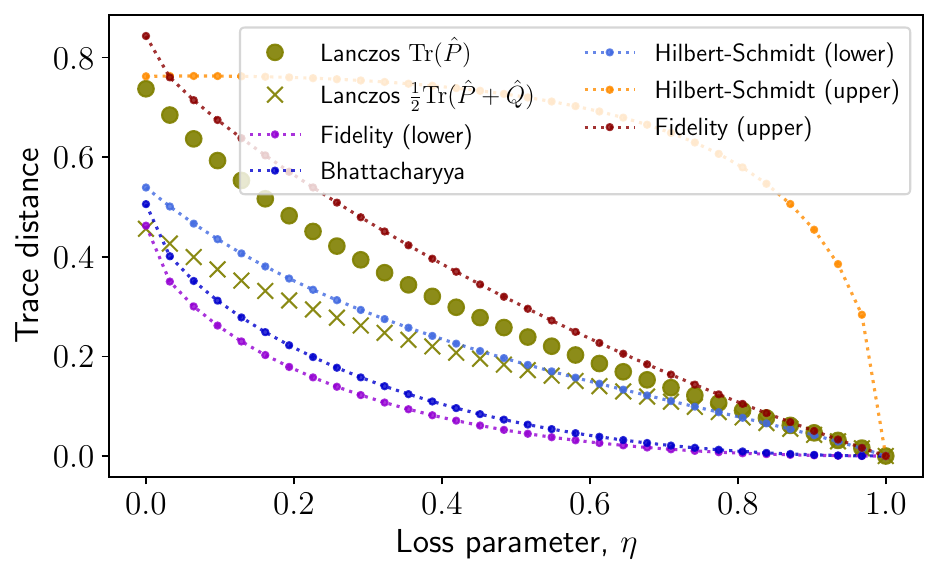}%
        }
        \caption{Illustration of the lower bounds on the trace distance obtained by extending Lanczos algorithm to the case in which both Gaussian states are mixed. \textit{(Top)} Trace distance between the states $\hat{\varrho}_{\pm}=\mathcal{L}_\eta[\hat{D}(\pm\alpha)\hat{S}(r)|0\rangle\langle 0|\hat{S}^\dagger(r)\hat{D}^\dagger(\pm\alpha)]$, where $\mathcal{L}_\eta[\cdot]$ is a loss channel with transmission $1-\eta$, $\hat{D}(\alpha)$ is a displacement operator, and $\hat{S}(r)$ is a squeezing operator. The parametrization of these Gaussian states is $\bm{r}_{\pm}=\pm\sqrt{2\hbar(1-\eta)}(\mathrm{Re}(\alpha), \mathrm{Im}(\alpha))$ and $\bm{V}_{\pm} = \frac{\hbar}{2}(1-\eta)\big(\begin{smallmatrix}e^{-2r}&0\\0&e^{2r}\end{smallmatrix}\big)+\frac{\hbar}{2}\eta\mathbb{I}_2$. The estimated results are shown as a function of $\eta\in[0.5, 1.0]$ for $\alpha = 0.8$, and $r\in\{0.05, 0.3, 1.5\}$ (or, equivalently, squeezing values of $\{0.43\,\text{dB}, 2.61\,\text{dB}, 13.03\,\text{dB}\}$). Solid lines correspond to the diagonalization of $\hat{\varrho}_+-\hat{\varrho}_-$ in Fock basis with cutoff $c=100$. Circles and crosses correspond to an estimation running Lanczos algorithm for $\ell = 5$ steps using a Gaussian trial vector with $\bm{r}=(1.5, 1.5)$ and $\bm{V}=\frac{\hbar}{2}\mathbb{I}_2$. \textit{(Bottom)} Trace distance between the $(M=5)$-mode mixed Gaussian state parametrized by $\bm{r}_1=\bm{0}$ and $\bm{V}_1=\frac{\hbar}{2}(\mathbb{I}_M)\oplus([(1-\eta)(1 + 4 \sinh^2(s)) + \eta]\mathbb{I}_M)$, and the mixed state represented by $\bm{r}_2=\bm{0}$ and $\bm{V}_2=\frac{\hbar}{2}([(1-\eta)e^{-2s}+\eta]\mathbb{I}_M)\oplus([(1-\eta)e^{2s}+\eta]\mathbb{I}_M)$. Results are presented as a function of the loss parameter $\eta\in[0.0,1.0]$, for $s=0.5$. Circles and crosses correspond to an estimation running Lanczos algorithm for $\ell = 4$ steps using a Gaussian trial vector with $\bm{r}=\bm{1}_{2M}$ and $\bm{V}=\frac{\hbar}{2}\mathbb{I}_{2M}$, where $\bm{1}_{2M}$ is a $2M$-length vector whose components are all equal to 1. Several bounds on the trace distance~\cite{weedbrook2012gaussian} are shown in dotted lines.}
        \label{fig:td_gaussian_mixed}
    \end{figure}

    \begin{figure}[!t]
        {
        \includegraphics[width=\columnwidth]{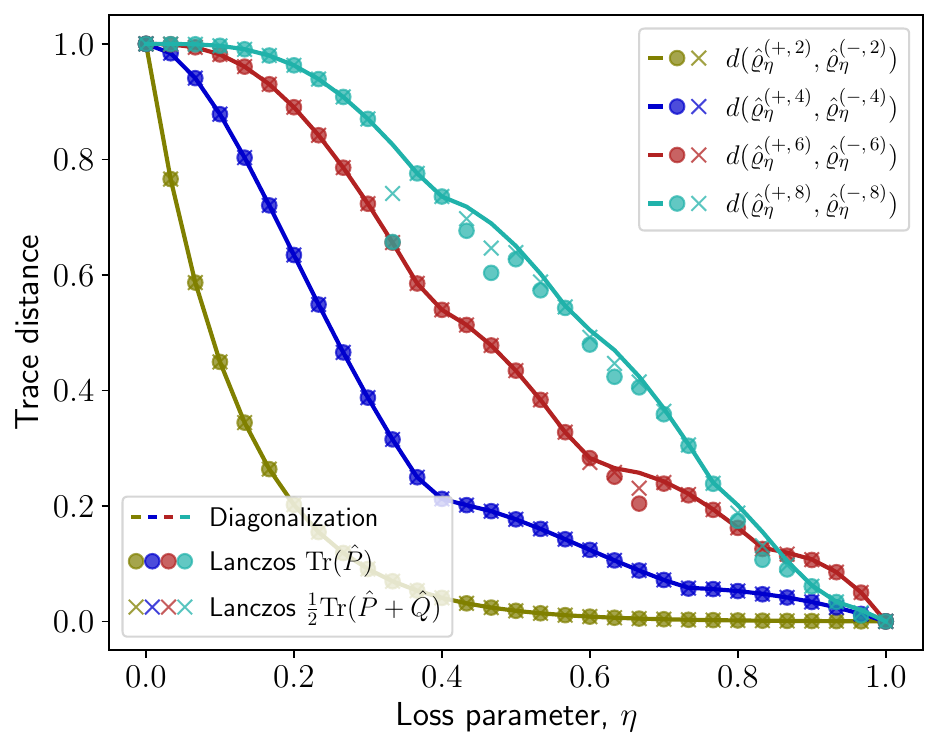}%
        }
        \caption{Illustration of the lower bounds on the trace distance obtained by extending Lanczos algorithm to the case in which both mixed states can be written as linear combinations of outer products of pure Gaussian states. We estimate the trace distance between the mixed states $\hat{\varrho}_\eta^{(\pm,\,p)}=\sum_{j,k=0}^{p-1}b_{j,k}^{\pm}|\sqrt{1-\eta}\alpha e^{i2\pi j/p}\rangle\langle \sqrt{1-\eta}\alpha e^{i2\pi k/p}|$, where $b_{j,k}^{\pm} = (\pm1)^{j+k}e^{-\eta|\alpha|^2}\exp(\eta|\alpha|^2e^{i2\pi (j-k)/p})/N_{\pm}$. Here, $|\alpha\rangle$ is a coherent state, $\bm{r}=\sqrt{2\hbar}(\mathrm{Re}(\alpha), \mathrm{Im}(\alpha))$ and $\bm{V}=\frac{\hbar}{2}\mathbb{I}_2$, $\eta$ is a loss parameter, and $N_{\pm}$ are normalization constants. The results are shown as a function of $\eta\in[0.0, 1.0]$ for $\alpha = 2.0$, and $p\in\{2, 4, 6, 8\}$. Solid lines correspond to the diagonalization of $\hat{\varrho}_\eta^{(+,\,p)}-\hat{\varrho}_\eta^{(-,\,p)}$ in Fock basis with cutoff $c=100$. Circles and crosses correspond to an estimation running Lanczos algorithm for $\ell = 10$ steps using $|\alpha\rangle$ as trial vector.}
        \label{fig:td_cat_mixed}
    \end{figure}
    
    To conclude our discussion, let us note that one can also use the Lanczos algorithm to compute $d(\hat{\varrho}_1,\hat{\varrho}_2)$, with $\hat{\varrho}_1$ and $\hat{\varrho}_2$ mixed Gaussian states,  by letting $\hat{A}=\Delta\hat{\varrho}=\hat{\varrho}_1-\hat{\varrho}_2$, and by choosing an appropriate Gaussian pure state $|c\rangle$. This choice leads to recursive relations similar to Eqs.~\eqref{eq:recurrence_matrix_C} to~\eqref{eq:tridiagonal_entries}, in which the entries of the metric take the form
    \begin{equation}
        \langle c|(\Delta\hat{\varrho})^\ell|c\rangle = \sum_{\bm{u}\in \{1,2\}^\ell}(-1)^{|\bm{u}|}\mathrm{Tr}(\hat{\varrho}_{u_1}\cdots\hat{\varrho}_{u_\ell}|c\rangle\langle c|),
        \label{eq:metric_mixed_states}
    \end{equation}
    where $|\bm{u}|=\sum_{k=1}^\ell(u_k-1)$. These entries can be computed using Bargmann invariants, but the construction of the whole metric will have a time of computation that grows as $\mathcal{O}(2^\ell)$. Lanczos algorithm approximates the eigenvalues at the edges of the spectrum very quickly, but needs to be run for a significant number of steps in order to obtain an accurate estimation of the remaining eigenvalues, which are needed to obtain either $\mathrm{Tr}(\hat{P})$ or $\frac{1}{2}\mathrm{Tr}(\hat{P}+\hat{Q})$. Thus, the computation of $d(\hat{\varrho}_1,\hat{\varrho}_2)$ using this technique becomes more complicated, and by running the algorithm for a reasonable number of steps one can only obtain a lower bound on the trace distance (see Fig.~\ref{fig:td_gaussian_mixed} for an illustration of this lower bound). This is because the magnitude of the eigenvalues is approximated from below~\cite{kreuzer1981lanczos}, and because we can only efficiently obtain a few of them.

    We can also consider the situation in which the mixed states $\hat{\varrho}_1$ and $\hat{\varrho}_2$ are linear combinations of outer products of pure Gaussian states, so that their difference takes the form $\Delta\hat{\varrho} = \sum_{j,k=1}^qb_{j,k}|f_j\rangle\langle f_k|$. In this case, $\langle c|(\Delta\hat{\varrho})^\ell|c\rangle$ can be computed using Eqs.~\eqref{eq:linear_comb_metric} and~\eqref{eq:recursive_power_expectation_coeffs}, replacing $|\psi\rangle$ by $|c\rangle$. This significantly reduces the time complexity of obtaining the metric. However, we still only obtain a lower bound on the trace distance because even if we obtain a large number of the eigenvalues needed, all of them will be approximated from below in absolute value, most of them (those far from the edges) without good enough accuracy (see Fig.~\ref{fig:td_cat_mixed} for an illustration of the lower bound in this situation). 
    
    \textit{Conclusion}--- In this work we have presented a numerical technique for the computation of the trace distance between a pure and a mixed Gaussian state. This method avoids the diagonalization of approximate matrix representations of density operators, depends only on the parametrization of Gaussian states in terms of covariance matrices and vectors of first moments, and its time complexity scales polynomially with the number of modes of the CV system under consideration. Our technique, which is based on the well-known Lanczos algorithm for the estimation of eigenvalues, significantly simplifies the cost of computation of trace distances between Gaussian states, and we believe it will become a widely used tool in the fields of quantum information processing and quantum computing.

    We also proved that our results can be extended to the case in which the states involved can be expressed as linear combinations of outer products of pure Gaussian states. This opens the door to the computation of trace distances for a wide range of bosonic states of interest. Future work will be devoted to investigating whether the algorithm can be adjusted to include more general combinations of Gaussian states.

    Finally, we have shown that our methods provide a lower bound on the trace distance between mixed Gaussian states. The challenge of turning this lower bound into a more accurate estimation of the trace distance will rely on developing more efficient algorithms for the computation of the corresponding metric, as well as on finding novel Krylov subspace techniques to better approximate the trace of the positive part of Hermitian operators. We will tackle these challenges in future work.

    \emph{Acknowledgments}---
        J.M.-C. and N.Q. acknowledge the support from the Ministère de l’Économie et de l’Innovation du Québec and the Natural Sciences and Engineering Research Council of Canada. They also thank S. Molesky and P. Virally for insightful discussions.
    
    \bibliography{bib.bib}

\begin{thebibliography}{69}%
\makeatletter
\providecommand \@ifxundefined [1]{%
 \@ifx{#1\undefined}
}%
\providecommand \@ifnum [1]{%
 \ifnum #1\expandafter \@firstoftwo
 \else \expandafter \@secondoftwo
 \fi
}%
\providecommand \@ifx [1]{%
 \ifx #1\expandafter \@firstoftwo
 \else \expandafter \@secondoftwo
 \fi
}%
\providecommand \natexlab [1]{#1}%
\providecommand \enquote  [1]{``#1''}%
\providecommand \bibnamefont  [1]{#1}%
\providecommand \bibfnamefont [1]{#1}%
\providecommand \citenamefont [1]{#1}%
\providecommand \href@noop [0]{\@secondoftwo}%
\providecommand \href [0]{\begingroup \@sanitize@url \@href}%
\providecommand \@href[1]{\@@startlink{#1}\@@href}%
\providecommand \@@href[1]{\endgroup#1\@@endlink}%
\providecommand \@sanitize@url [0]{\catcode `\\12\catcode `\$12\catcode `\&12\catcode `\#12\catcode `\^12\catcode `\_12\catcode `\%12\relax}%
\providecommand \@@startlink[1]{}%
\providecommand \@@endlink[0]{}%
\providecommand \url  [0]{\begingroup\@sanitize@url \@url }%
\providecommand \@url [1]{\endgroup\@href {#1}{\urlprefix }}%
\providecommand \urlprefix  [0]{URL }%
\providecommand \Eprint [0]{\href }%
\providecommand \doibase [0]{https://doi.org/}%
\providecommand \selectlanguage [0]{\@gobble}%
\providecommand \bibinfo  [0]{\@secondoftwo}%
\providecommand \bibfield  [0]{\@secondoftwo}%
\providecommand \translation [1]{[#1]}%
\providecommand \BibitemOpen [0]{}%
\providecommand \bibitemStop [0]{}%
\providecommand \bibitemNoStop [0]{.\EOS\space}%
\providecommand \EOS [0]{\spacefactor3000\relax}%
\providecommand \BibitemShut  [1]{\csname bibitem#1\endcsname}%
\let\auto@bib@innerbib\@empty
\bibitem [{\citenamefont {Serafini}(2023)}]{serafini2023quantum}%
  \BibitemOpen
  \bibfield  {author} {\bibinfo {author} {\bibfnamefont {A.}~\bibnamefont {Serafini}},\ }\href@noop {} {\emph {\bibinfo {title} {Quantum continuous variables: a primer of theoretical methods}}}\ (\bibinfo  {publisher} {CRC press},\ \bibinfo {year} {2023})\BibitemShut {NoStop}%
\bibitem [{\citenamefont {Cerf}\ \emph {et~al.}(2007)\citenamefont {Cerf}, \citenamefont {Leuchs},\ and\ \citenamefont {Polzik}}]{cerf2007quantum}%
  \BibitemOpen
  \bibfield  {author} {\bibinfo {author} {\bibfnamefont {N.~J.}\ \bibnamefont {Cerf}}, \bibinfo {author} {\bibfnamefont {G.}~\bibnamefont {Leuchs}},\ and\ \bibinfo {author} {\bibfnamefont {E.~S.}\ \bibnamefont {Polzik}},\ }\href@noop {} {\emph {\bibinfo {title} {Quantum information with continuous variables of atoms and light}}}\ (\bibinfo  {publisher} {World Scientific},\ \bibinfo {year} {2007})\BibitemShut {NoStop}%
\bibitem [{\citenamefont {Usenko}\ \emph {et~al.}(2025)\citenamefont {Usenko}, \citenamefont {Ac{\'\i}n}, \citenamefont {All{\'e}aume}, \citenamefont {Andersen}, \citenamefont {Diamanti}, \citenamefont {Gehring}, \citenamefont {Hajomer}, \citenamefont {Kanitschar}, \citenamefont {Pacher}, \citenamefont {Pirandola} \emph {et~al.}}]{usenko2025continuous}%
  \BibitemOpen
  \bibfield  {author} {\bibinfo {author} {\bibfnamefont {V.~C.}\ \bibnamefont {Usenko}}, \bibinfo {author} {\bibfnamefont {A.}~\bibnamefont {Ac{\'\i}n}}, \bibinfo {author} {\bibfnamefont {R.}~\bibnamefont {All{\'e}aume}}, \bibinfo {author} {\bibfnamefont {U.~L.}\ \bibnamefont {Andersen}}, \bibinfo {author} {\bibfnamefont {E.}~\bibnamefont {Diamanti}}, \bibinfo {author} {\bibfnamefont {T.}~\bibnamefont {Gehring}}, \bibinfo {author} {\bibfnamefont {A.~A.}\ \bibnamefont {Hajomer}}, \bibinfo {author} {\bibfnamefont {F.}~\bibnamefont {Kanitschar}}, \bibinfo {author} {\bibfnamefont {C.}~\bibnamefont {Pacher}}, \bibinfo {author} {\bibfnamefont {S.}~\bibnamefont {Pirandola}}, \emph {et~al.},\ }\bibfield  {title} {\bibinfo {title} {Continuous-variable quantum communication},\ }\href@noop {} {\bibfield  {journal} {\bibinfo  {journal} {arXiv preprint arXiv:2501.12801}\ } (\bibinfo {year} {2025})}\BibitemShut {NoStop}%
\bibitem [{\citenamefont {Jain}\ \emph {et~al.}(2022)\citenamefont {Jain}, \citenamefont {Chin}, \citenamefont {Mani}, \citenamefont {Lupo}, \citenamefont {Nikolic}, \citenamefont {Kordts}, \citenamefont {Pirandola}, \citenamefont {Pedersen}, \citenamefont {Kolb}, \citenamefont {{\"O}mer} \emph {et~al.}}]{jain2022practical}%
  \BibitemOpen
  \bibfield  {author} {\bibinfo {author} {\bibfnamefont {N.}~\bibnamefont {Jain}}, \bibinfo {author} {\bibfnamefont {H.-M.}\ \bibnamefont {Chin}}, \bibinfo {author} {\bibfnamefont {H.}~\bibnamefont {Mani}}, \bibinfo {author} {\bibfnamefont {C.}~\bibnamefont {Lupo}}, \bibinfo {author} {\bibfnamefont {D.~S.}\ \bibnamefont {Nikolic}}, \bibinfo {author} {\bibfnamefont {A.}~\bibnamefont {Kordts}}, \bibinfo {author} {\bibfnamefont {S.}~\bibnamefont {Pirandola}}, \bibinfo {author} {\bibfnamefont {T.~B.}\ \bibnamefont {Pedersen}}, \bibinfo {author} {\bibfnamefont {M.}~\bibnamefont {Kolb}}, \bibinfo {author} {\bibfnamefont {B.}~\bibnamefont {{\"O}mer}}, \emph {et~al.},\ }\bibfield  {title} {\bibinfo {title} {Practical continuous-variable quantum key distribution with composable security},\ }\href@noop {} {\bibfield  {journal} {\bibinfo  {journal} {Nature communications}\ }\textbf {\bibinfo {volume} {13}},\ \bibinfo {pages} {4740} (\bibinfo {year} {2022})}\BibitemShut {NoStop}%
\bibitem [{\citenamefont {Madsen}\ \emph {et~al.}(2012)\citenamefont {Madsen}, \citenamefont {Usenko}, \citenamefont {Lassen}, \citenamefont {Filip},\ and\ \citenamefont {Andersen}}]{madsen2012continuous}%
  \BibitemOpen
  \bibfield  {author} {\bibinfo {author} {\bibfnamefont {L.~S.}\ \bibnamefont {Madsen}}, \bibinfo {author} {\bibfnamefont {V.~C.}\ \bibnamefont {Usenko}}, \bibinfo {author} {\bibfnamefont {M.}~\bibnamefont {Lassen}}, \bibinfo {author} {\bibfnamefont {R.}~\bibnamefont {Filip}},\ and\ \bibinfo {author} {\bibfnamefont {U.~L.}\ \bibnamefont {Andersen}},\ }\bibfield  {title} {\bibinfo {title} {Continuous variable quantum key distribution with modulated entangled states},\ }\href@noop {} {\bibfield  {journal} {\bibinfo  {journal} {Nature communications}\ }\textbf {\bibinfo {volume} {3}},\ \bibinfo {pages} {1083} (\bibinfo {year} {2012})}\BibitemShut {NoStop}%
\bibitem [{\citenamefont {Vaidman}(1994)}]{vaidman1994teleportation}%
  \BibitemOpen
  \bibfield  {author} {\bibinfo {author} {\bibfnamefont {L.}~\bibnamefont {Vaidman}},\ }\bibfield  {title} {\bibinfo {title} {Teleportation of quantum states},\ }\href@noop {} {\bibfield  {journal} {\bibinfo  {journal} {Phys. Rev. A}\ }\textbf {\bibinfo {volume} {49}},\ \bibinfo {pages} {1473} (\bibinfo {year} {1994})}\BibitemShut {NoStop}%
\bibitem [{\citenamefont {Fadel}\ \emph {et~al.}(2025)\citenamefont {Fadel}, \citenamefont {Roux},\ and\ \citenamefont {Gessner}}]{fadel2025quantum}%
  \BibitemOpen
  \bibfield  {author} {\bibinfo {author} {\bibfnamefont {M.}~\bibnamefont {Fadel}}, \bibinfo {author} {\bibfnamefont {N.}~\bibnamefont {Roux}},\ and\ \bibinfo {author} {\bibfnamefont {M.}~\bibnamefont {Gessner}},\ }\bibfield  {title} {\bibinfo {title} {Quantum metrology with a continuous-variable system},\ }\href@noop {} {\bibfield  {journal} {\bibinfo  {journal} {Reports on Progress in Physics}\ }\textbf {\bibinfo {volume} {88}},\ \bibinfo {pages} {106001} (\bibinfo {year} {2025})}\BibitemShut {NoStop}%
\bibitem [{\citenamefont {Aasi}\ \emph {et~al.}(2013)\citenamefont {Aasi}, \citenamefont {Abadie}, \citenamefont {Abbott}, \citenamefont {Abbott}, \citenamefont {Abbott}, \citenamefont {Abernathy}, \citenamefont {Adams}, \citenamefont {Adams}, \citenamefont {Addesso}, \citenamefont {Adhikari} \emph {et~al.}}]{aasi2013enhanced}%
  \BibitemOpen
  \bibfield  {author} {\bibinfo {author} {\bibfnamefont {J.}~\bibnamefont {Aasi}}, \bibinfo {author} {\bibfnamefont {J.}~\bibnamefont {Abadie}}, \bibinfo {author} {\bibfnamefont {B.}~\bibnamefont {Abbott}}, \bibinfo {author} {\bibfnamefont {R.}~\bibnamefont {Abbott}}, \bibinfo {author} {\bibfnamefont {T.}~\bibnamefont {Abbott}}, \bibinfo {author} {\bibfnamefont {M.}~\bibnamefont {Abernathy}}, \bibinfo {author} {\bibfnamefont {C.}~\bibnamefont {Adams}}, \bibinfo {author} {\bibfnamefont {T.}~\bibnamefont {Adams}}, \bibinfo {author} {\bibfnamefont {P.}~\bibnamefont {Addesso}}, \bibinfo {author} {\bibfnamefont {R.}~\bibnamefont {Adhikari}}, \emph {et~al.},\ }\bibfield  {title} {\bibinfo {title} {Enhanced sensitivity of the ligo gravitational wave detector by using squeezed states of light},\ }\href@noop {} {\bibfield  {journal} {\bibinfo  {journal} {Nature Photonics}\ }\textbf {\bibinfo {volume} {7}},\ \bibinfo {pages} {613} (\bibinfo {year} {2013})}\BibitemShut {NoStop}%
\bibitem [{\citenamefont {Meyer}\ \emph {et~al.}(2001)\citenamefont {Meyer}, \citenamefont {Rowe}, \citenamefont {Kielpinski}, \citenamefont {Sackett}, \citenamefont {Itano}, \citenamefont {Monroe},\ and\ \citenamefont {Wineland}}]{meyer2001experimental}%
  \BibitemOpen
  \bibfield  {author} {\bibinfo {author} {\bibfnamefont {V.}~\bibnamefont {Meyer}}, \bibinfo {author} {\bibfnamefont {M.}~\bibnamefont {Rowe}}, \bibinfo {author} {\bibfnamefont {D.}~\bibnamefont {Kielpinski}}, \bibinfo {author} {\bibfnamefont {C.}~\bibnamefont {Sackett}}, \bibinfo {author} {\bibfnamefont {W.~M.}\ \bibnamefont {Itano}}, \bibinfo {author} {\bibfnamefont {C.}~\bibnamefont {Monroe}},\ and\ \bibinfo {author} {\bibfnamefont {D.~J.}\ \bibnamefont {Wineland}},\ }\bibfield  {title} {\bibinfo {title} {Experimental demonstration of entanglement-enhanced rotation angle estimation using trapped ions},\ }\href@noop {} {\bibfield  {journal} {\bibinfo  {journal} {Physical review letters}\ }\textbf {\bibinfo {volume} {86}},\ \bibinfo {pages} {5870} (\bibinfo {year} {2001})}\BibitemShut {NoStop}%
\bibitem [{\citenamefont {McCormick}\ \emph {et~al.}(2019)\citenamefont {McCormick}, \citenamefont {Keller}, \citenamefont {Burd}, \citenamefont {Wineland}, \citenamefont {Wilson},\ and\ \citenamefont {Leibfried}}]{mccormick2019quantum}%
  \BibitemOpen
  \bibfield  {author} {\bibinfo {author} {\bibfnamefont {K.~C.}\ \bibnamefont {McCormick}}, \bibinfo {author} {\bibfnamefont {J.}~\bibnamefont {Keller}}, \bibinfo {author} {\bibfnamefont {S.~C.}\ \bibnamefont {Burd}}, \bibinfo {author} {\bibfnamefont {D.~J.}\ \bibnamefont {Wineland}}, \bibinfo {author} {\bibfnamefont {A.~C.}\ \bibnamefont {Wilson}},\ and\ \bibinfo {author} {\bibfnamefont {D.}~\bibnamefont {Leibfried}},\ }\bibfield  {title} {\bibinfo {title} {Quantum-enhanced sensing of a single-ion mechanical oscillator},\ }\href@noop {} {\bibfield  {journal} {\bibinfo  {journal} {Nature}\ }\textbf {\bibinfo {volume} {572}},\ \bibinfo {pages} {86} (\bibinfo {year} {2019})}\BibitemShut {NoStop}%
\bibitem [{\citenamefont {Braunstein}\ and\ \citenamefont {Van~Loock}(2005)}]{braunstein2005quantum}%
  \BibitemOpen
  \bibfield  {author} {\bibinfo {author} {\bibfnamefont {S.~L.}\ \bibnamefont {Braunstein}}\ and\ \bibinfo {author} {\bibfnamefont {P.}~\bibnamefont {Van~Loock}},\ }\bibfield  {title} {\bibinfo {title} {Quantum information with continuous variables},\ }\href@noop {} {\bibfield  {journal} {\bibinfo  {journal} {Reviews of modern physics}\ }\textbf {\bibinfo {volume} {77}},\ \bibinfo {pages} {513} (\bibinfo {year} {2005})}\BibitemShut {NoStop}%
\bibitem [{\citenamefont {Weedbrook}\ \emph {et~al.}(2012)\citenamefont {Weedbrook}, \citenamefont {Pirandola}, \citenamefont {Garc{\'\i}a-Patr{\'o}n}, \citenamefont {Cerf}, \citenamefont {Ralph}, \citenamefont {Shapiro},\ and\ \citenamefont {Lloyd}}]{weedbrook2012gaussian}%
  \BibitemOpen
  \bibfield  {author} {\bibinfo {author} {\bibfnamefont {C.}~\bibnamefont {Weedbrook}}, \bibinfo {author} {\bibfnamefont {S.}~\bibnamefont {Pirandola}}, \bibinfo {author} {\bibfnamefont {R.}~\bibnamefont {Garc{\'\i}a-Patr{\'o}n}}, \bibinfo {author} {\bibfnamefont {N.~J.}\ \bibnamefont {Cerf}}, \bibinfo {author} {\bibfnamefont {T.~C.}\ \bibnamefont {Ralph}}, \bibinfo {author} {\bibfnamefont {J.~H.}\ \bibnamefont {Shapiro}},\ and\ \bibinfo {author} {\bibfnamefont {S.}~\bibnamefont {Lloyd}},\ }\bibfield  {title} {\bibinfo {title} {Gaussian quantum information},\ }\href@noop {} {\bibfield  {journal} {\bibinfo  {journal} {Reviews of Modern Physics}\ }\textbf {\bibinfo {volume} {84}},\ \bibinfo {pages} {621} (\bibinfo {year} {2012})}\BibitemShut {NoStop}%
\bibitem [{\citenamefont {Andersen}\ \emph {et~al.}(2010)\citenamefont {Andersen}, \citenamefont {Leuchs},\ and\ \citenamefont {Silberhorn}}]{andersen2010continuous}%
  \BibitemOpen
  \bibfield  {author} {\bibinfo {author} {\bibfnamefont {U.~L.}\ \bibnamefont {Andersen}}, \bibinfo {author} {\bibfnamefont {G.}~\bibnamefont {Leuchs}},\ and\ \bibinfo {author} {\bibfnamefont {C.}~\bibnamefont {Silberhorn}},\ }\bibfield  {title} {\bibinfo {title} {Continuous-variable quantum information processing},\ }\href@noop {} {\bibfield  {journal} {\bibinfo  {journal} {Laser \& Photonics Reviews}\ }\textbf {\bibinfo {volume} {4}},\ \bibinfo {pages} {337} (\bibinfo {year} {2010})}\BibitemShut {NoStop}%
\bibitem [{\citenamefont {Asavanant}\ and\ \citenamefont {Furusawa}(2022)}]{asavanant2022optical}%
  \BibitemOpen
  \bibfield  {author} {\bibinfo {author} {\bibfnamefont {W.}~\bibnamefont {Asavanant}}\ and\ \bibinfo {author} {\bibfnamefont {A.}~\bibnamefont {Furusawa}},\ }\href@noop {} {\emph {\bibinfo {title} {Optical Quantum Computers: A Route to Practical Continuous Variable Quantum Information Processing}}}\ (\bibinfo  {publisher} {AIP Publishing LLC},\ \bibinfo {year} {2022})\BibitemShut {NoStop}%
\bibitem [{\citenamefont {{PsiQuantum Team}}(2025)}]{psiquantum2025manufacturable}%
  \BibitemOpen
  \bibfield  {author} {\bibinfo {author} {\bibnamefont {{PsiQuantum Team}}},\ }\bibfield  {title} {\bibinfo {title} {A manufacturable platform for photonic quantum computing},\ }\href@noop {} {\bibfield  {journal} {\bibinfo  {journal} {Nature}\ }\textbf {\bibinfo {volume} {641}},\ \bibinfo {pages} {876} (\bibinfo {year} {2025})}\BibitemShut {NoStop}%
\bibitem [{\citenamefont {Bartolucci}\ \emph {et~al.}(2023)\citenamefont {Bartolucci}, \citenamefont {Birchall}, \citenamefont {Bombin}, \citenamefont {Cable}, \citenamefont {Dawson}, \citenamefont {Gimeno-Segovia}, \citenamefont {Johnston}, \citenamefont {Kieling}, \citenamefont {Nickerson}, \citenamefont {Pant} \emph {et~al.}}]{bartolucci2023fusion}%
  \BibitemOpen
  \bibfield  {author} {\bibinfo {author} {\bibfnamefont {S.}~\bibnamefont {Bartolucci}}, \bibinfo {author} {\bibfnamefont {P.}~\bibnamefont {Birchall}}, \bibinfo {author} {\bibfnamefont {H.}~\bibnamefont {Bombin}}, \bibinfo {author} {\bibfnamefont {H.}~\bibnamefont {Cable}}, \bibinfo {author} {\bibfnamefont {C.}~\bibnamefont {Dawson}}, \bibinfo {author} {\bibfnamefont {M.}~\bibnamefont {Gimeno-Segovia}}, \bibinfo {author} {\bibfnamefont {E.}~\bibnamefont {Johnston}}, \bibinfo {author} {\bibfnamefont {K.}~\bibnamefont {Kieling}}, \bibinfo {author} {\bibfnamefont {N.}~\bibnamefont {Nickerson}}, \bibinfo {author} {\bibfnamefont {M.}~\bibnamefont {Pant}}, \emph {et~al.},\ }\bibfield  {title} {\bibinfo {title} {Fusion-based quantum computation},\ }\href@noop {} {\bibfield  {journal} {\bibinfo  {journal} {Nature Communications}\ }\textbf {\bibinfo {volume} {14}},\ \bibinfo {pages} {912} (\bibinfo {year} {2023})}\BibitemShut {NoStop}%
\bibitem [{\citenamefont {Bourassa}\ \emph {et~al.}(2021{\natexlab{a}})\citenamefont {Bourassa}, \citenamefont {Alexander}, \citenamefont {Vasmer}, \citenamefont {Patil}, \citenamefont {Tzitrin}, \citenamefont {Matsuura}, \citenamefont {Su}, \citenamefont {Baragiola}, \citenamefont {Guha}, \citenamefont {Dauphinais} \emph {et~al.}}]{bourassa2021blueprint}%
  \BibitemOpen
  \bibfield  {author} {\bibinfo {author} {\bibfnamefont {J.~E.}\ \bibnamefont {Bourassa}}, \bibinfo {author} {\bibfnamefont {R.~N.}\ \bibnamefont {Alexander}}, \bibinfo {author} {\bibfnamefont {M.}~\bibnamefont {Vasmer}}, \bibinfo {author} {\bibfnamefont {A.}~\bibnamefont {Patil}}, \bibinfo {author} {\bibfnamefont {I.}~\bibnamefont {Tzitrin}}, \bibinfo {author} {\bibfnamefont {T.}~\bibnamefont {Matsuura}}, \bibinfo {author} {\bibfnamefont {D.}~\bibnamefont {Su}}, \bibinfo {author} {\bibfnamefont {B.~Q.}\ \bibnamefont {Baragiola}}, \bibinfo {author} {\bibfnamefont {S.}~\bibnamefont {Guha}}, \bibinfo {author} {\bibfnamefont {G.}~\bibnamefont {Dauphinais}}, \emph {et~al.},\ }\bibfield  {title} {\bibinfo {title} {Blueprint for a scalable photonic fault-tolerant quantum computer},\ }\href@noop {} {\bibfield  {journal} {\bibinfo  {journal} {Quantum}\ }\textbf {\bibinfo {volume} {5}},\ \bibinfo {pages} {392} (\bibinfo {year} {2021}{\natexlab{a}})}\BibitemShut {NoStop}%
\bibitem [{\citenamefont {Andersen}\ \emph {et~al.}(2015)\citenamefont {Andersen}, \citenamefont {Neergaard-Nielsen}, \citenamefont {Van~Loock},\ and\ \citenamefont {Furusawa}}]{andersen2015hybrid}%
  \BibitemOpen
  \bibfield  {author} {\bibinfo {author} {\bibfnamefont {U.~L.}\ \bibnamefont {Andersen}}, \bibinfo {author} {\bibfnamefont {J.~S.}\ \bibnamefont {Neergaard-Nielsen}}, \bibinfo {author} {\bibfnamefont {P.}~\bibnamefont {Van~Loock}},\ and\ \bibinfo {author} {\bibfnamefont {A.}~\bibnamefont {Furusawa}},\ }\bibfield  {title} {\bibinfo {title} {Hybrid discrete-and continuous-variable quantum information},\ }\href@noop {} {\bibfield  {journal} {\bibinfo  {journal} {Nature Physics}\ }\textbf {\bibinfo {volume} {11}},\ \bibinfo {pages} {713} (\bibinfo {year} {2015})}\BibitemShut {NoStop}%
\bibitem [{\citenamefont {Menicucci}(2014)}]{menicucci2014fault}%
  \BibitemOpen
  \bibfield  {author} {\bibinfo {author} {\bibfnamefont {N.~C.}\ \bibnamefont {Menicucci}},\ }\bibfield  {title} {\bibinfo {title} {Fault-tolerant measurement-based quantum computing with continuous-variable cluster states},\ }\href@noop {} {\bibfield  {journal} {\bibinfo  {journal} {Physical review letters}\ }\textbf {\bibinfo {volume} {112}},\ \bibinfo {pages} {120504} (\bibinfo {year} {2014})}\BibitemShut {NoStop}%
\bibitem [{\citenamefont {Knill}\ \emph {et~al.}(2001)\citenamefont {Knill}, \citenamefont {Laflamme},\ and\ \citenamefont {Milburn}}]{knill2001scheme}%
  \BibitemOpen
  \bibfield  {author} {\bibinfo {author} {\bibfnamefont {E.}~\bibnamefont {Knill}}, \bibinfo {author} {\bibfnamefont {R.}~\bibnamefont {Laflamme}},\ and\ \bibinfo {author} {\bibfnamefont {G.~J.}\ \bibnamefont {Milburn}},\ }\bibfield  {title} {\bibinfo {title} {A scheme for efficient quantum computation with linear optics},\ }\href@noop {} {\bibfield  {journal} {\bibinfo  {journal} {nature}\ }\textbf {\bibinfo {volume} {409}},\ \bibinfo {pages} {46} (\bibinfo {year} {2001})}\BibitemShut {NoStop}%
\bibitem [{\citenamefont {Menicucci}\ \emph {et~al.}(2006)\citenamefont {Menicucci}, \citenamefont {Van~Loock}, \citenamefont {Gu}, \citenamefont {Weedbrook}, \citenamefont {Ralph},\ and\ \citenamefont {Nielsen}}]{menicucci2006universal}%
  \BibitemOpen
  \bibfield  {author} {\bibinfo {author} {\bibfnamefont {N.~C.}\ \bibnamefont {Menicucci}}, \bibinfo {author} {\bibfnamefont {P.}~\bibnamefont {Van~Loock}}, \bibinfo {author} {\bibfnamefont {M.}~\bibnamefont {Gu}}, \bibinfo {author} {\bibfnamefont {C.}~\bibnamefont {Weedbrook}}, \bibinfo {author} {\bibfnamefont {.~f. T.~C.}\ \bibnamefont {Ralph}},\ and\ \bibinfo {author} {\bibfnamefont {M.~A.}\ \bibnamefont {Nielsen}},\ }\bibfield  {title} {\bibinfo {title} {Universal quantum computation with continuous-variable cluster states},\ }\href@noop {} {\bibfield  {journal} {\bibinfo  {journal} {Physical review letters}\ }\textbf {\bibinfo {volume} {97}},\ \bibinfo {pages} {110501} (\bibinfo {year} {2006})}\BibitemShut {NoStop}%
\bibitem [{\citenamefont {Nielsen}(2004)}]{nielsen2004optical}%
  \BibitemOpen
  \bibfield  {author} {\bibinfo {author} {\bibfnamefont {M.~A.}\ \bibnamefont {Nielsen}},\ }\bibfield  {title} {\bibinfo {title} {Optical quantum computation using cluster states},\ }\href@noop {} {\bibfield  {journal} {\bibinfo  {journal} {Physical review letters}\ }\textbf {\bibinfo {volume} {93}},\ \bibinfo {pages} {040503} (\bibinfo {year} {2004})}\BibitemShut {NoStop}%
\bibitem [{\citenamefont {Gottesman}\ \emph {et~al.}(2001)\citenamefont {Gottesman}, \citenamefont {Kitaev},\ and\ \citenamefont {Preskill}}]{gottesman2001encoding}%
  \BibitemOpen
  \bibfield  {author} {\bibinfo {author} {\bibfnamefont {D.}~\bibnamefont {Gottesman}}, \bibinfo {author} {\bibfnamefont {A.}~\bibnamefont {Kitaev}},\ and\ \bibinfo {author} {\bibfnamefont {J.}~\bibnamefont {Preskill}},\ }\bibfield  {title} {\bibinfo {title} {Encoding a qubit in an oscillator},\ }\href@noop {} {\bibfield  {journal} {\bibinfo  {journal} {Phys. Rev. A}\ }\textbf {\bibinfo {volume} {64}},\ \bibinfo {pages} {012310} (\bibinfo {year} {2001})}\BibitemShut {NoStop}%
\bibitem [{\citenamefont {Lloyd}\ and\ \citenamefont {Braunstein}(1999)}]{lloyd1999quantum}%
  \BibitemOpen
  \bibfield  {author} {\bibinfo {author} {\bibfnamefont {S.}~\bibnamefont {Lloyd}}\ and\ \bibinfo {author} {\bibfnamefont {S.~L.}\ \bibnamefont {Braunstein}},\ }\bibfield  {title} {\bibinfo {title} {Quantum computation over continuous variables},\ }\href@noop {} {\bibfield  {journal} {\bibinfo  {journal} {Physical Review Letters}\ }\textbf {\bibinfo {volume} {82}},\ \bibinfo {pages} {1784} (\bibinfo {year} {1999})}\BibitemShut {NoStop}%
\bibitem [{\citenamefont {Aaronson}\ and\ \citenamefont {Arkhipov}(2011)}]{aaronson2011computational}%
  \BibitemOpen
  \bibfield  {author} {\bibinfo {author} {\bibfnamefont {S.}~\bibnamefont {Aaronson}}\ and\ \bibinfo {author} {\bibfnamefont {A.}~\bibnamefont {Arkhipov}},\ }\bibfield  {title} {\bibinfo {title} {The computational complexity of linear optics},\ }in\ \href@noop {} {\emph {\bibinfo {booktitle} {Proceedings of the forty-third annual ACM symposium on Theory of computing}}}\ (\bibinfo {year} {2011})\ pp.\ \bibinfo {pages} {333--342}\BibitemShut {NoStop}%
\bibitem [{\citenamefont {Hamilton}\ \emph {et~al.}(2017)\citenamefont {Hamilton}, \citenamefont {Kruse}, \citenamefont {Sansoni}, \citenamefont {Barkhofen}, \citenamefont {Silberhorn},\ and\ \citenamefont {Jex}}]{hamilton2017gaussian}%
  \BibitemOpen
  \bibfield  {author} {\bibinfo {author} {\bibfnamefont {C.~S.}\ \bibnamefont {Hamilton}}, \bibinfo {author} {\bibfnamefont {R.}~\bibnamefont {Kruse}}, \bibinfo {author} {\bibfnamefont {L.}~\bibnamefont {Sansoni}}, \bibinfo {author} {\bibfnamefont {S.}~\bibnamefont {Barkhofen}}, \bibinfo {author} {\bibfnamefont {C.}~\bibnamefont {Silberhorn}},\ and\ \bibinfo {author} {\bibfnamefont {I.}~\bibnamefont {Jex}},\ }\bibfield  {title} {\bibinfo {title} {Gaussian boson sampling},\ }\href@noop {} {\bibfield  {journal} {\bibinfo  {journal} {Physical review letters}\ }\textbf {\bibinfo {volume} {119}},\ \bibinfo {pages} {170501} (\bibinfo {year} {2017})}\BibitemShut {NoStop}%
\bibitem [{\citenamefont {Madsen}\ \emph {et~al.}(2022)\citenamefont {Madsen}, \citenamefont {Laudenbach}, \citenamefont {Askarani}, \citenamefont {Rortais}, \citenamefont {Vincent}, \citenamefont {Bulmer}, \citenamefont {Miatto}, \citenamefont {Neuhaus}, \citenamefont {Helt}, \citenamefont {Collins} \emph {et~al.}}]{madsen2022quantum}%
  \BibitemOpen
  \bibfield  {author} {\bibinfo {author} {\bibfnamefont {L.~S.}\ \bibnamefont {Madsen}}, \bibinfo {author} {\bibfnamefont {F.}~\bibnamefont {Laudenbach}}, \bibinfo {author} {\bibfnamefont {M.~F.}\ \bibnamefont {Askarani}}, \bibinfo {author} {\bibfnamefont {F.}~\bibnamefont {Rortais}}, \bibinfo {author} {\bibfnamefont {T.}~\bibnamefont {Vincent}}, \bibinfo {author} {\bibfnamefont {J.~F.}\ \bibnamefont {Bulmer}}, \bibinfo {author} {\bibfnamefont {F.~M.}\ \bibnamefont {Miatto}}, \bibinfo {author} {\bibfnamefont {L.}~\bibnamefont {Neuhaus}}, \bibinfo {author} {\bibfnamefont {L.~G.}\ \bibnamefont {Helt}}, \bibinfo {author} {\bibfnamefont {M.~J.}\ \bibnamefont {Collins}}, \emph {et~al.},\ }\bibfield  {title} {\bibinfo {title} {Quantum computational advantage with a programmable photonic processor},\ }\href@noop {} {\bibfield  {journal} {\bibinfo  {journal} {Nature}\ }\textbf {\bibinfo {volume} {606}},\ \bibinfo {pages} {75} (\bibinfo {year} {2022})}\BibitemShut {NoStop}%
\bibitem [{\citenamefont {Zhong}\ \emph {et~al.}(2020)\citenamefont {Zhong}, \citenamefont {Wang}, \citenamefont {Deng}, \citenamefont {Chen}, \citenamefont {Peng}, \citenamefont {Luo}, \citenamefont {Qin}, \citenamefont {Wu}, \citenamefont {Ding}, \citenamefont {Hu} \emph {et~al.}}]{zhong2020quantum}%
  \BibitemOpen
  \bibfield  {author} {\bibinfo {author} {\bibfnamefont {H.-S.}\ \bibnamefont {Zhong}}, \bibinfo {author} {\bibfnamefont {H.}~\bibnamefont {Wang}}, \bibinfo {author} {\bibfnamefont {Y.-H.}\ \bibnamefont {Deng}}, \bibinfo {author} {\bibfnamefont {M.-C.}\ \bibnamefont {Chen}}, \bibinfo {author} {\bibfnamefont {L.-C.}\ \bibnamefont {Peng}}, \bibinfo {author} {\bibfnamefont {Y.-H.}\ \bibnamefont {Luo}}, \bibinfo {author} {\bibfnamefont {J.}~\bibnamefont {Qin}}, \bibinfo {author} {\bibfnamefont {D.}~\bibnamefont {Wu}}, \bibinfo {author} {\bibfnamefont {X.}~\bibnamefont {Ding}}, \bibinfo {author} {\bibfnamefont {Y.}~\bibnamefont {Hu}}, \emph {et~al.},\ }\bibfield  {title} {\bibinfo {title} {Quantum computational advantage using photons},\ }\href@noop {} {\bibfield  {journal} {\bibinfo  {journal} {Science}\ }\textbf {\bibinfo {volume} {370}},\ \bibinfo {pages} {1460} (\bibinfo {year} {2020})}\BibitemShut {NoStop}%
\bibitem [{\citenamefont {Zhong}\ \emph {et~al.}(2021)\citenamefont {Zhong}, \citenamefont {Deng}, \citenamefont {Qin}, \citenamefont {Wang}, \citenamefont {Chen}, \citenamefont {Peng}, \citenamefont {Luo}, \citenamefont {Wu}, \citenamefont {Gong}, \citenamefont {Su} \emph {et~al.}}]{zhong2021phase}%
  \BibitemOpen
  \bibfield  {author} {\bibinfo {author} {\bibfnamefont {H.-S.}\ \bibnamefont {Zhong}}, \bibinfo {author} {\bibfnamefont {Y.-H.}\ \bibnamefont {Deng}}, \bibinfo {author} {\bibfnamefont {J.}~\bibnamefont {Qin}}, \bibinfo {author} {\bibfnamefont {H.}~\bibnamefont {Wang}}, \bibinfo {author} {\bibfnamefont {M.-C.}\ \bibnamefont {Chen}}, \bibinfo {author} {\bibfnamefont {L.-C.}\ \bibnamefont {Peng}}, \bibinfo {author} {\bibfnamefont {Y.-H.}\ \bibnamefont {Luo}}, \bibinfo {author} {\bibfnamefont {D.}~\bibnamefont {Wu}}, \bibinfo {author} {\bibfnamefont {S.-Q.}\ \bibnamefont {Gong}}, \bibinfo {author} {\bibfnamefont {H.}~\bibnamefont {Su}}, \emph {et~al.},\ }\bibfield  {title} {\bibinfo {title} {Phase-programmable gaussian boson sampling using stimulated squeezed light},\ }\href@noop {} {\bibfield  {journal} {\bibinfo  {journal} {Physical review letters}\ }\textbf {\bibinfo {volume} {127}},\ \bibinfo {pages} {180502} (\bibinfo {year} {2021})}\BibitemShut {NoStop}%
\bibitem [{\citenamefont {Deng}\ \emph {et~al.}(2023)\citenamefont {Deng}, \citenamefont {Gu}, \citenamefont {Liu}, \citenamefont {Gong}, \citenamefont {Su}, \citenamefont {Zhang}, \citenamefont {Tang}, \citenamefont {Jia}, \citenamefont {Xu}, \citenamefont {Chen} \emph {et~al.}}]{deng2023gaussian}%
  \BibitemOpen
  \bibfield  {author} {\bibinfo {author} {\bibfnamefont {Y.-H.}\ \bibnamefont {Deng}}, \bibinfo {author} {\bibfnamefont {Y.-C.}\ \bibnamefont {Gu}}, \bibinfo {author} {\bibfnamefont {H.-L.}\ \bibnamefont {Liu}}, \bibinfo {author} {\bibfnamefont {S.-Q.}\ \bibnamefont {Gong}}, \bibinfo {author} {\bibfnamefont {H.}~\bibnamefont {Su}}, \bibinfo {author} {\bibfnamefont {Z.-J.}\ \bibnamefont {Zhang}}, \bibinfo {author} {\bibfnamefont {H.-Y.}\ \bibnamefont {Tang}}, \bibinfo {author} {\bibfnamefont {M.-H.}\ \bibnamefont {Jia}}, \bibinfo {author} {\bibfnamefont {J.-M.}\ \bibnamefont {Xu}}, \bibinfo {author} {\bibfnamefont {M.-C.}\ \bibnamefont {Chen}}, \emph {et~al.},\ }\bibfield  {title} {\bibinfo {title} {Gaussian boson sampling with pseudo-photon-number-resolving detectors and quantum computational advantage},\ }\href@noop {} {\bibfield  {journal} {\bibinfo  {journal} {Physical review letters}\ }\textbf {\bibinfo {volume} {131}},\ \bibinfo {pages} {150601} (\bibinfo {year} {2023})}\BibitemShut {NoStop}%
\bibitem [{\citenamefont {Liu}\ \emph {et~al.}(2025)\citenamefont {Liu}, \citenamefont {Su}, \citenamefont {Gong}, \citenamefont {Gu}, \citenamefont {Tang}, \citenamefont {Jia}, \citenamefont {Wei}, \citenamefont {Song}, \citenamefont {Wang}, \citenamefont {Zheng} \emph {et~al.}}]{liu2025robust}%
  \BibitemOpen
  \bibfield  {author} {\bibinfo {author} {\bibfnamefont {H.-L.}\ \bibnamefont {Liu}}, \bibinfo {author} {\bibfnamefont {H.}~\bibnamefont {Su}}, \bibinfo {author} {\bibfnamefont {S.-Q.}\ \bibnamefont {Gong}}, \bibinfo {author} {\bibfnamefont {Y.-C.}\ \bibnamefont {Gu}}, \bibinfo {author} {\bibfnamefont {H.-Y.}\ \bibnamefont {Tang}}, \bibinfo {author} {\bibfnamefont {M.-H.}\ \bibnamefont {Jia}}, \bibinfo {author} {\bibfnamefont {Q.}~\bibnamefont {Wei}}, \bibinfo {author} {\bibfnamefont {Y.}~\bibnamefont {Song}}, \bibinfo {author} {\bibfnamefont {D.}~\bibnamefont {Wang}}, \bibinfo {author} {\bibfnamefont {M.}~\bibnamefont {Zheng}}, \emph {et~al.},\ }\bibfield  {title} {\bibinfo {title} {Robust quantum computational advantage with programmable 3050-photon gaussian boson sampling},\ }\href@noop {} {\bibfield  {journal} {\bibinfo  {journal} {arXiv preprint arXiv:2508.09092}\ } (\bibinfo {year} {2025})}\BibitemShut {NoStop}%
\bibitem [{\citenamefont {Ferraro}\ \emph {et~al.}(2005)\citenamefont {Ferraro}, \citenamefont {Olivares},\ and\ \citenamefont {Paris}}]{ferraro2005gaussian}%
  \BibitemOpen
  \bibfield  {author} {\bibinfo {author} {\bibfnamefont {A.}~\bibnamefont {Ferraro}}, \bibinfo {author} {\bibfnamefont {S.}~\bibnamefont {Olivares}},\ and\ \bibinfo {author} {\bibfnamefont {M.~G.}\ \bibnamefont {Paris}},\ }\bibfield  {title} {\bibinfo {title} {Gaussian states in continuous variable quantum information},\ }\href@noop {} {\bibfield  {journal} {\bibinfo  {journal} {arXiv preprint quant-ph/0503237}\ } (\bibinfo {year} {2005})}\BibitemShut {NoStop}%
\bibitem [{\citenamefont {Barnett}\ and\ \citenamefont {Radmore}(2002)}]{barnett2002methods}%
  \BibitemOpen
  \bibfield  {author} {\bibinfo {author} {\bibfnamefont {S.}~\bibnamefont {Barnett}}\ and\ \bibinfo {author} {\bibfnamefont {P.~M.}\ \bibnamefont {Radmore}},\ }\href@noop {} {\emph {\bibinfo {title} {Methods in theoretical quantum optics}}},\ Vol.~\bibinfo {volume} {15}\ (\bibinfo  {publisher} {Oxford University Press},\ \bibinfo {year} {2002})\BibitemShut {NoStop}%
\bibitem [{\citenamefont {Quesada}\ \emph {et~al.}(2022)\citenamefont {Quesada}, \citenamefont {Helt}, \citenamefont {Menotti}, \citenamefont {Liscidini},\ and\ \citenamefont {Sipe}}]{quesada2022beyond}%
  \BibitemOpen
  \bibfield  {author} {\bibinfo {author} {\bibfnamefont {N.}~\bibnamefont {Quesada}}, \bibinfo {author} {\bibfnamefont {L.}~\bibnamefont {Helt}}, \bibinfo {author} {\bibfnamefont {M.}~\bibnamefont {Menotti}}, \bibinfo {author} {\bibfnamefont {M.}~\bibnamefont {Liscidini}},\ and\ \bibinfo {author} {\bibfnamefont {J.}~\bibnamefont {Sipe}},\ }\bibfield  {title} {\bibinfo {title} {Beyond photon pairs—nonlinear quantum photonics in the high-gain regime: a tutorial},\ }\href@noop {} {\bibfield  {journal} {\bibinfo  {journal} {Advances in Optics and Photonics}\ }\textbf {\bibinfo {volume} {14}},\ \bibinfo {pages} {291} (\bibinfo {year} {2022})}\BibitemShut {NoStop}%
\bibitem [{\citenamefont {Walschaers}(2021)}]{walschaers2021non}%
  \BibitemOpen
  \bibfield  {author} {\bibinfo {author} {\bibfnamefont {M.}~\bibnamefont {Walschaers}},\ }\bibfield  {title} {\bibinfo {title} {Non-gaussian quantum states and where to find them},\ }\href@noop {} {\bibfield  {journal} {\bibinfo  {journal} {PRX quantum}\ }\textbf {\bibinfo {volume} {2}},\ \bibinfo {pages} {030204} (\bibinfo {year} {2021})}\BibitemShut {NoStop}%
\bibitem [{\citenamefont {Tiedau}\ \emph {et~al.}(2019)\citenamefont {Tiedau}, \citenamefont {Bartley}, \citenamefont {Harder}, \citenamefont {Lita}, \citenamefont {Nam}, \citenamefont {Gerrits},\ and\ \citenamefont {Silberhorn}}]{tiedau2019scalability}%
  \BibitemOpen
  \bibfield  {author} {\bibinfo {author} {\bibfnamefont {J.}~\bibnamefont {Tiedau}}, \bibinfo {author} {\bibfnamefont {T.~J.}\ \bibnamefont {Bartley}}, \bibinfo {author} {\bibfnamefont {G.}~\bibnamefont {Harder}}, \bibinfo {author} {\bibfnamefont {A.~E.}\ \bibnamefont {Lita}}, \bibinfo {author} {\bibfnamefont {S.~W.}\ \bibnamefont {Nam}}, \bibinfo {author} {\bibfnamefont {T.}~\bibnamefont {Gerrits}},\ and\ \bibinfo {author} {\bibfnamefont {C.}~\bibnamefont {Silberhorn}},\ }\bibfield  {title} {\bibinfo {title} {Scalability of parametric down-conversion for generating higher-order fock states},\ }\href@noop {} {\bibfield  {journal} {\bibinfo  {journal} {Physical Review A}\ }\textbf {\bibinfo {volume} {100}},\ \bibinfo {pages} {041802} (\bibinfo {year} {2019})}\BibitemShut {NoStop}%
\bibitem [{\citenamefont {Konno}\ \emph {et~al.}(2024)\citenamefont {Konno}, \citenamefont {Asavanant}, \citenamefont {Hanamura}, \citenamefont {Nagayoshi}, \citenamefont {Fukui}, \citenamefont {Sakaguchi}, \citenamefont {Ide}, \citenamefont {China}, \citenamefont {Yabuno}, \citenamefont {Miki} \emph {et~al.}}]{konno2024logical}%
  \BibitemOpen
  \bibfield  {author} {\bibinfo {author} {\bibfnamefont {S.}~\bibnamefont {Konno}}, \bibinfo {author} {\bibfnamefont {W.}~\bibnamefont {Asavanant}}, \bibinfo {author} {\bibfnamefont {F.}~\bibnamefont {Hanamura}}, \bibinfo {author} {\bibfnamefont {H.}~\bibnamefont {Nagayoshi}}, \bibinfo {author} {\bibfnamefont {K.}~\bibnamefont {Fukui}}, \bibinfo {author} {\bibfnamefont {A.}~\bibnamefont {Sakaguchi}}, \bibinfo {author} {\bibfnamefont {R.}~\bibnamefont {Ide}}, \bibinfo {author} {\bibfnamefont {F.}~\bibnamefont {China}}, \bibinfo {author} {\bibfnamefont {M.}~\bibnamefont {Yabuno}}, \bibinfo {author} {\bibfnamefont {S.}~\bibnamefont {Miki}}, \emph {et~al.},\ }\bibfield  {title} {\bibinfo {title} {Logical states for fault-tolerant quantum computation with propagating light},\ }\href@noop {} {\bibfield  {journal} {\bibinfo  {journal} {Science}\ }\textbf {\bibinfo {volume} {383}},\ \bibinfo {pages} {289} (\bibinfo {year} {2024})}\BibitemShut {NoStop}%
\bibitem [{\citenamefont {Helstrom}(1969)}]{helstrom1969quantum}%
  \BibitemOpen
  \bibfield  {author} {\bibinfo {author} {\bibfnamefont {C.~W.}\ \bibnamefont {Helstrom}},\ }\bibfield  {title} {\bibinfo {title} {Quantum detection and estimation theory},\ }\href@noop {} {\bibfield  {journal} {\bibinfo  {journal} {Journal of Statistical Physics}\ }\textbf {\bibinfo {volume} {1}},\ \bibinfo {pages} {231} (\bibinfo {year} {1969})}\BibitemShut {NoStop}%
\bibitem [{\citenamefont {Mele}\ \emph {et~al.}(2025{\natexlab{a}})\citenamefont {Mele}, \citenamefont {Mele}, \citenamefont {Bittel}, \citenamefont {Eisert}, \citenamefont {Giovannetti}, \citenamefont {Lami}, \citenamefont {Leone},\ and\ \citenamefont {Oliviero}}]{mele2025learning}%
  \BibitemOpen
  \bibfield  {author} {\bibinfo {author} {\bibfnamefont {F.~A.}\ \bibnamefont {Mele}}, \bibinfo {author} {\bibfnamefont {A.~A.}\ \bibnamefont {Mele}}, \bibinfo {author} {\bibfnamefont {L.}~\bibnamefont {Bittel}}, \bibinfo {author} {\bibfnamefont {J.}~\bibnamefont {Eisert}}, \bibinfo {author} {\bibfnamefont {V.}~\bibnamefont {Giovannetti}}, \bibinfo {author} {\bibfnamefont {L.}~\bibnamefont {Lami}}, \bibinfo {author} {\bibfnamefont {L.}~\bibnamefont {Leone}},\ and\ \bibinfo {author} {\bibfnamefont {S.~F.}\ \bibnamefont {Oliviero}},\ }\bibfield  {title} {\bibinfo {title} {Learning quantum states of continuous-variable systems},\ }\href@noop {} {\bibfield  {journal} {\bibinfo  {journal} {Nature Physics}\ ,\ \bibinfo {pages} {1}} (\bibinfo {year} {2025}{\natexlab{a}})}\BibitemShut {NoStop}%
\bibitem [{\citenamefont {Anshu}\ and\ \citenamefont {Arunachalam}(2024)}]{anshu2024survey}%
  \BibitemOpen
  \bibfield  {author} {\bibinfo {author} {\bibfnamefont {A.}~\bibnamefont {Anshu}}\ and\ \bibinfo {author} {\bibfnamefont {S.}~\bibnamefont {Arunachalam}},\ }\bibfield  {title} {\bibinfo {title} {A survey on the complexity of learning quantum states},\ }\href@noop {} {\bibfield  {journal} {\bibinfo  {journal} {Nature Reviews Physics}\ }\textbf {\bibinfo {volume} {6}},\ \bibinfo {pages} {59} (\bibinfo {year} {2024})}\BibitemShut {NoStop}%
\bibitem [{\citenamefont {Bittel}\ \emph {et~al.}(2025{\natexlab{a}})\citenamefont {Bittel}, \citenamefont {Mele}, \citenamefont {Mele}, \citenamefont {Tirone},\ and\ \citenamefont {Lami}}]{bittel2025optimal}%
  \BibitemOpen
  \bibfield  {author} {\bibinfo {author} {\bibfnamefont {L.}~\bibnamefont {Bittel}}, \bibinfo {author} {\bibfnamefont {F.~A.}\ \bibnamefont {Mele}}, \bibinfo {author} {\bibfnamefont {A.~A.}\ \bibnamefont {Mele}}, \bibinfo {author} {\bibfnamefont {S.}~\bibnamefont {Tirone}},\ and\ \bibinfo {author} {\bibfnamefont {L.}~\bibnamefont {Lami}},\ }\bibfield  {title} {\bibinfo {title} {Optimal estimates of trace distance between bosonic gaussian states and applications to learning},\ }\href@noop {} {\bibfield  {journal} {\bibinfo  {journal} {Quantum}\ }\textbf {\bibinfo {volume} {9}},\ \bibinfo {pages} {1769} (\bibinfo {year} {2025}{\natexlab{a}})}\BibitemShut {NoStop}%
\bibitem [{\citenamefont {Mele}\ \emph {et~al.}(2025{\natexlab{b}})\citenamefont {Mele}, \citenamefont {Oliviero}, \citenamefont {Upreti},\ and\ \citenamefont {Chabaud}}]{mele2025symplectic}%
  \BibitemOpen
  \bibfield  {author} {\bibinfo {author} {\bibfnamefont {F.~A.}\ \bibnamefont {Mele}}, \bibinfo {author} {\bibfnamefont {S.~F.~E.}\ \bibnamefont {Oliviero}}, \bibinfo {author} {\bibfnamefont {V.}~\bibnamefont {Upreti}},\ and\ \bibinfo {author} {\bibfnamefont {U.}~\bibnamefont {Chabaud}},\ }\bibfield  {title} {\bibinfo {title} {The symplectic rank of non-gaussian quantum states},\ }\href@noop {} {\bibfield  {journal} {\bibinfo  {journal} {arXiv preprint arXiv:2504.19319}\ } (\bibinfo {year} {2025}{\natexlab{b}})}\BibitemShut {NoStop}%
\bibitem [{\citenamefont {Bittel}\ \emph {et~al.}(2025{\natexlab{b}})\citenamefont {Bittel}, \citenamefont {Mele}, \citenamefont {Eisert},\ and\ \citenamefont {Mele}}]{bittel2025energy}%
  \BibitemOpen
  \bibfield  {author} {\bibinfo {author} {\bibfnamefont {L.}~\bibnamefont {Bittel}}, \bibinfo {author} {\bibfnamefont {F.~A.}\ \bibnamefont {Mele}}, \bibinfo {author} {\bibfnamefont {J.}~\bibnamefont {Eisert}},\ and\ \bibinfo {author} {\bibfnamefont {A.~A.}\ \bibnamefont {Mele}},\ }\bibfield  {title} {\bibinfo {title} {Energy-independent tomography of gaussian states},\ }\href@noop {} {\bibfield  {journal} {\bibinfo  {journal} {arXiv preprint arXiv:2508.14979}\ } (\bibinfo {year} {2025}{\natexlab{b}})}\BibitemShut {NoStop}%
\bibitem [{\citenamefont {Holevo}(2024)}]{holevo2024estimates}%
  \BibitemOpen
  \bibfield  {author} {\bibinfo {author} {\bibfnamefont {A.~S.}\ \bibnamefont {Holevo}},\ }\bibfield  {title} {\bibinfo {title} {On estimates of trace-norm distance between quantum gaussian states},\ }\href@noop {} {\bibfield  {journal} {\bibinfo  {journal} {arXiv preprint arXiv:2408.11400}\ } (\bibinfo {year} {2024})}\BibitemShut {NoStop}%
\bibitem [{\citenamefont {van Luijk}\ \emph {et~al.}(2024)\citenamefont {van Luijk}, \citenamefont {Galke}, \citenamefont {Hahn},\ and\ \citenamefont {Burgarth}}]{van2024error}%
  \BibitemOpen
  \bibfield  {author} {\bibinfo {author} {\bibfnamefont {L.}~\bibnamefont {van Luijk}}, \bibinfo {author} {\bibfnamefont {N.}~\bibnamefont {Galke}}, \bibinfo {author} {\bibfnamefont {A.}~\bibnamefont {Hahn}},\ and\ \bibinfo {author} {\bibfnamefont {D.}~\bibnamefont {Burgarth}},\ }\bibfield  {title} {\bibinfo {title} {Error bounds for lie group representations in quantum mechanics},\ }\href@noop {} {\bibfield  {journal} {\bibinfo  {journal} {Journal of Physics A: Mathematical and Theoretical}\ }\textbf {\bibinfo {volume} {57}},\ \bibinfo {pages} {105301} (\bibinfo {year} {2024})}\BibitemShut {NoStop}%
\bibitem [{\citenamefont {Becker}\ \emph {et~al.}(2021)\citenamefont {Becker}, \citenamefont {Datta}, \citenamefont {Lami},\ and\ \citenamefont {Rouz{\'e}}}]{becker2021energy}%
  \BibitemOpen
  \bibfield  {author} {\bibinfo {author} {\bibfnamefont {S.}~\bibnamefont {Becker}}, \bibinfo {author} {\bibfnamefont {N.}~\bibnamefont {Datta}}, \bibinfo {author} {\bibfnamefont {L.}~\bibnamefont {Lami}},\ and\ \bibinfo {author} {\bibfnamefont {C.}~\bibnamefont {Rouz{\'e}}},\ }\bibfield  {title} {\bibinfo {title} {Energy-constrained discrimination of unitaries, quantum speed limits, and a gaussian solovay-kitaev theorem},\ }\href@noop {} {\bibfield  {journal} {\bibinfo  {journal} {Physical Review Letters}\ }\textbf {\bibinfo {volume} {126}},\ \bibinfo {pages} {190504} (\bibinfo {year} {2021})}\BibitemShut {NoStop}%
\bibitem [{\citenamefont {Bulmer}\ \emph {et~al.}(2026)\citenamefont {Bulmer}, \citenamefont {Mart{\'\i}nez-Cifuentes}, \citenamefont {Bell},\ and\ \citenamefont {Quesada}}]{bulmer2026simulating}%
  \BibitemOpen
  \bibfield  {author} {\bibinfo {author} {\bibfnamefont {J.~F.}\ \bibnamefont {Bulmer}}, \bibinfo {author} {\bibfnamefont {J.}~\bibnamefont {Mart{\'\i}nez-Cifuentes}}, \bibinfo {author} {\bibfnamefont {B.~A.}\ \bibnamefont {Bell}},\ and\ \bibinfo {author} {\bibfnamefont {N.}~\bibnamefont {Quesada}},\ }\bibfield  {title} {\bibinfo {title} {Simulating lossy and partially distinguishable quantum optical circuits: theory, algorithms, and applications to experiment validation and state preparation},\ }\href@noop {} {\bibfield  {journal} {\bibinfo  {journal} {Advanced Photonics}\ }\textbf {\bibinfo {volume} {8}},\ \bibinfo {pages} {016010} (\bibinfo {year} {2026})}\BibitemShut {NoStop}%
\bibitem [{\citenamefont {Quesada}\ and\ \citenamefont {Arrazola}(2020)}]{quesada2020exact}%
  \BibitemOpen
  \bibfield  {author} {\bibinfo {author} {\bibfnamefont {N.}~\bibnamefont {Quesada}}\ and\ \bibinfo {author} {\bibfnamefont {J.~M.}\ \bibnamefont {Arrazola}},\ }\bibfield  {title} {\bibinfo {title} {Exact simulation of gaussian boson sampling in polynomial space and exponential time},\ }\href@noop {} {\bibfield  {journal} {\bibinfo  {journal} {Physical Review Research}\ }\textbf {\bibinfo {volume} {2}},\ \bibinfo {pages} {023005} (\bibinfo {year} {2020})}\BibitemShut {NoStop}%
\bibitem [{\citenamefont {Qi}\ \emph {et~al.}(2022)\citenamefont {Qi}, \citenamefont {Cifuentes}, \citenamefont {Br{\'a}dler}, \citenamefont {Israel}, \citenamefont {Kalajdzievski},\ and\ \citenamefont {Quesada}}]{qi2022efficient}%
  \BibitemOpen
  \bibfield  {author} {\bibinfo {author} {\bibfnamefont {H.}~\bibnamefont {Qi}}, \bibinfo {author} {\bibfnamefont {D.}~\bibnamefont {Cifuentes}}, \bibinfo {author} {\bibfnamefont {K.}~\bibnamefont {Br{\'a}dler}}, \bibinfo {author} {\bibfnamefont {R.}~\bibnamefont {Israel}}, \bibinfo {author} {\bibfnamefont {T.}~\bibnamefont {Kalajdzievski}},\ and\ \bibinfo {author} {\bibfnamefont {N.}~\bibnamefont {Quesada}},\ }\bibfield  {title} {\bibinfo {title} {Efficient sampling from shallow gaussian quantum-optical circuits with local interactions},\ }\href@noop {} {\bibfield  {journal} {\bibinfo  {journal} {Physical Review A}\ }\textbf {\bibinfo {volume} {105}},\ \bibinfo {pages} {052412} (\bibinfo {year} {2022})}\BibitemShut {NoStop}%
\bibitem [{\citenamefont {Mele}\ \emph {et~al.}(2025{\natexlab{c}})\citenamefont {Mele}, \citenamefont {Barbarino}, \citenamefont {Giovannetti},\ and\ \citenamefont {Fanizza}}]{mele2025achievable}%
  \BibitemOpen
  \bibfield  {author} {\bibinfo {author} {\bibfnamefont {F.~A.}\ \bibnamefont {Mele}}, \bibinfo {author} {\bibfnamefont {G.}~\bibnamefont {Barbarino}}, \bibinfo {author} {\bibfnamefont {V.}~\bibnamefont {Giovannetti}},\ and\ \bibinfo {author} {\bibfnamefont {M.}~\bibnamefont {Fanizza}},\ }\bibfield  {title} {\bibinfo {title} {Achievable rates in non-asymptotic bosonic quantum communication},\ }\href@noop {} {\bibfield  {journal} {\bibinfo  {journal} {arXiv preprint arXiv:2502.05524}\ } (\bibinfo {year} {2025}{\natexlab{c}})}\BibitemShut {NoStop}%
\bibitem [{\citenamefont {Kreuzer}\ \emph {et~al.}(1981)\citenamefont {Kreuzer}, \citenamefont {Miller},\ and\ \citenamefont {Berger}}]{kreuzer1981lanczos}%
  \BibitemOpen
  \bibfield  {author} {\bibinfo {author} {\bibfnamefont {K.}~\bibnamefont {Kreuzer}}, \bibinfo {author} {\bibfnamefont {H.}~\bibnamefont {Miller}},\ and\ \bibinfo {author} {\bibfnamefont {W.}~\bibnamefont {Berger}},\ }\bibfield  {title} {\bibinfo {title} {The lanczos algorithm for self-adjoint operators},\ }\href@noop {} {\bibfield  {journal} {\bibinfo  {journal} {Physics Letters A}\ }\textbf {\bibinfo {volume} {81}},\ \bibinfo {pages} {429} (\bibinfo {year} {1981})}\BibitemShut {NoStop}%
\bibitem [{\citenamefont {Saad}(2003)}]{saad2003iterative}%
  \BibitemOpen
  \bibfield  {author} {\bibinfo {author} {\bibfnamefont {Y.}~\bibnamefont {Saad}},\ }\href@noop {} {\emph {\bibinfo {title} {Iterative methods for sparse linear systems}}}\ (\bibinfo  {publisher} {SIAM},\ \bibinfo {year} {2003})\BibitemShut {NoStop}%
\bibitem [{\citenamefont {Caruso}\ \emph {et~al.}(2021)\citenamefont {Caruso}, \citenamefont {Michelangeli} \emph {et~al.}}]{caruso2021inverse}%
  \BibitemOpen
  \bibfield  {author} {\bibinfo {author} {\bibfnamefont {N.~A.}\ \bibnamefont {Caruso}}, \bibinfo {author} {\bibfnamefont {A.}~\bibnamefont {Michelangeli}}, \emph {et~al.},\ }\href@noop {} {\emph {\bibinfo {title} {Inverse linear problems on Hilbert space and their Krylov solvability}}}\ (\bibinfo  {publisher} {Springer},\ \bibinfo {year} {2021})\BibitemShut {NoStop}%
\bibitem [{Note1()}]{Note1}%
  \BibitemOpen
  \bibinfo {note} {Since $\lambda _+ = \protect \mathrm {Tr}(\protect \hat {P}) = \protect \mathrm {Tr}(\protect \hat {Q})$, it follows that $\lambda _+$ is equal to the sum of the magnitudes of all negative eigenvalues of $|\psi \rangle \langle \psi |-\protect \hat {\varrho }$, which implies that $\lambda _+$ is greater than or equal to each of these absolute values}\BibitemShut {NoStop}%
\bibitem [{\citenamefont {Mart{\'\i}nez-Cifuentes}\ \emph {et~al.}(2023)\citenamefont {Mart{\'\i}nez-Cifuentes}, \citenamefont {Fonseca-Romero},\ and\ \citenamefont {Quesada}}]{martinez2023classical}%
  \BibitemOpen
  \bibfield  {author} {\bibinfo {author} {\bibfnamefont {J.}~\bibnamefont {Mart{\'\i}nez-Cifuentes}}, \bibinfo {author} {\bibfnamefont {K.~M.}\ \bibnamefont {Fonseca-Romero}},\ and\ \bibinfo {author} {\bibfnamefont {N.}~\bibnamefont {Quesada}},\ }\bibfield  {title} {\bibinfo {title} {Classical models may be a better explanation of the jiuzhang 1.0 gaussian boson sampler than its targeted squeezed light model},\ }\href@noop {} {\bibfield  {journal} {\bibinfo  {journal} {Quantum}\ }\textbf {\bibinfo {volume} {7}},\ \bibinfo {pages} {1076} (\bibinfo {year} {2023})}\BibitemShut {NoStop}%
\bibitem [{\citenamefont {Xu}(2025)}]{xu2025bargmann}%
  \BibitemOpen
  \bibfield  {author} {\bibinfo {author} {\bibfnamefont {J.}~\bibnamefont {Xu}},\ }\bibfield  {title} {\bibinfo {title} {Bargmann invariants of gaussian states},\ }\href@noop {} {\bibfield  {journal} {\bibinfo  {journal} {arXiv preprint arXiv:2508.07155}\ } (\bibinfo {year} {2025})}\BibitemShut {NoStop}%
\bibitem [{Note2()}]{Note2}%
  \BibitemOpen
  \bibinfo {note} {This implementation is available at \protect \url {https://github.com/polyquantique/trace_distance_lczs}}\BibitemShut {NoStop}%
\bibitem [{\citenamefont {Bourassa}\ \emph {et~al.}(2021{\natexlab{b}})\citenamefont {Bourassa}, \citenamefont {Quesada}, \citenamefont {Tzitrin}, \citenamefont {Sz{\'a}va}, \citenamefont {Isacsson}, \citenamefont {Izaac}, \citenamefont {Sabapathy}, \citenamefont {Dauphinais},\ and\ \citenamefont {Dhand}}]{bourassa2021fast}%
  \BibitemOpen
  \bibfield  {author} {\bibinfo {author} {\bibfnamefont {J.~E.}\ \bibnamefont {Bourassa}}, \bibinfo {author} {\bibfnamefont {N.}~\bibnamefont {Quesada}}, \bibinfo {author} {\bibfnamefont {I.}~\bibnamefont {Tzitrin}}, \bibinfo {author} {\bibfnamefont {A.}~\bibnamefont {Sz{\'a}va}}, \bibinfo {author} {\bibfnamefont {T.}~\bibnamefont {Isacsson}}, \bibinfo {author} {\bibfnamefont {J.}~\bibnamefont {Izaac}}, \bibinfo {author} {\bibfnamefont {K.~K.}\ \bibnamefont {Sabapathy}}, \bibinfo {author} {\bibfnamefont {G.}~\bibnamefont {Dauphinais}},\ and\ \bibinfo {author} {\bibfnamefont {I.}~\bibnamefont {Dhand}},\ }\bibfield  {title} {\bibinfo {title} {Fast simulation of bosonic qubits via gaussian functions in phase space},\ }\href@noop {} {\bibfield  {journal} {\bibinfo  {journal} {PRX Quantum}\ }\textbf {\bibinfo {volume} {2}},\ \bibinfo {pages} {040315} (\bibinfo {year} {2021}{\natexlab{b}})}\BibitemShut {NoStop}%
\bibitem [{\citenamefont {Solodovnikova}\ \emph {et~al.}(2025)\citenamefont {Solodovnikova}, \citenamefont {Andersen},\ and\ \citenamefont {Neergaard-Nielsen}}]{solodovnikova2025fast}%
  \BibitemOpen
  \bibfield  {author} {\bibinfo {author} {\bibfnamefont {O.}~\bibnamefont {Solodovnikova}}, \bibinfo {author} {\bibfnamefont {U.~L.}\ \bibnamefont {Andersen}},\ and\ \bibinfo {author} {\bibfnamefont {J.~S.}\ \bibnamefont {Neergaard-Nielsen}},\ }\bibfield  {title} {\bibinfo {title} {Fast simulations of continuous-variable circuits using the coherent state decomposition},\ }\href@noop {} {\bibfield  {journal} {\bibinfo  {journal} {arXiv preprint arXiv:2508.06175}\ } (\bibinfo {year} {2025})}\BibitemShut {NoStop}%
\bibitem [{\citenamefont {Albert}\ \emph {et~al.}(2018)\citenamefont {Albert}, \citenamefont {Noh}, \citenamefont {Duivenvoorden}, \citenamefont {Young}, \citenamefont {Brierley}, \citenamefont {Reinhold}, \citenamefont {Vuillot}, \citenamefont {Li}, \citenamefont {Shen}, \citenamefont {Girvin} \emph {et~al.}}]{albert2018performance}%
  \BibitemOpen
  \bibfield  {author} {\bibinfo {author} {\bibfnamefont {V.~V.}\ \bibnamefont {Albert}}, \bibinfo {author} {\bibfnamefont {K.}~\bibnamefont {Noh}}, \bibinfo {author} {\bibfnamefont {K.}~\bibnamefont {Duivenvoorden}}, \bibinfo {author} {\bibfnamefont {D.~J.}\ \bibnamefont {Young}}, \bibinfo {author} {\bibfnamefont {R.}~\bibnamefont {Brierley}}, \bibinfo {author} {\bibfnamefont {P.}~\bibnamefont {Reinhold}}, \bibinfo {author} {\bibfnamefont {C.}~\bibnamefont {Vuillot}}, \bibinfo {author} {\bibfnamefont {L.}~\bibnamefont {Li}}, \bibinfo {author} {\bibfnamefont {C.}~\bibnamefont {Shen}}, \bibinfo {author} {\bibfnamefont {S.~M.}\ \bibnamefont {Girvin}}, \emph {et~al.},\ }\bibfield  {title} {\bibinfo {title} {Performance and structure of single-mode bosonic codes},\ }\href@noop {} {\bibfield  {journal} {\bibinfo  {journal} {Physical Review A}\ }\textbf {\bibinfo {volume} {97}},\ \bibinfo {pages} {032346} (\bibinfo {year} {2018})}\BibitemShut {NoStop}%
\bibitem [{\citenamefont {Dias}\ and\ \citenamefont {K{\"o}nig}(2024)}]{dias2024classical}%
  \BibitemOpen
  \bibfield  {author} {\bibinfo {author} {\bibfnamefont {B.}~\bibnamefont {Dias}}\ and\ \bibinfo {author} {\bibfnamefont {R.}~\bibnamefont {K{\"o}nig}},\ }\bibfield  {title} {\bibinfo {title} {Classical simulation of non-gaussian bosonic circuits},\ }\href@noop {} {\bibfield  {journal} {\bibinfo  {journal} {Physical Review A}\ }\textbf {\bibinfo {volume} {110}},\ \bibinfo {pages} {042402} (\bibinfo {year} {2024})}\BibitemShut {NoStop}%
\bibitem [{\citenamefont {Quesada}(2025)}]{quesada2025s}%
  \BibitemOpen
  \bibfield  {author} {\bibinfo {author} {\bibfnamefont {N.}~\bibnamefont {Quesada}},\ }\bibfield  {title} {\bibinfo {title} {What’s my phase again? computing the vacuum-to-vacuum amplitude of quadratic bosonic evolution},\ }\href@noop {} {\bibfield  {journal} {\bibinfo  {journal} {Proceedings of the Royal Society A: Mathematical, Physical and Engineering Sciences}\ }\textbf {\bibinfo {volume} {481}} (\bibinfo {year} {2025})}\BibitemShut {NoStop}%
\bibitem [{\citenamefont {Yao}\ \emph {et~al.}(2024)\citenamefont {Yao}, \citenamefont {Miatto},\ and\ \citenamefont {Quesada}}]{yao2024riemannian}%
  \BibitemOpen
  \bibfield  {author} {\bibinfo {author} {\bibfnamefont {Y.}~\bibnamefont {Yao}}, \bibinfo {author} {\bibfnamefont {F.}~\bibnamefont {Miatto}},\ and\ \bibinfo {author} {\bibfnamefont {N.}~\bibnamefont {Quesada}},\ }\bibfield  {title} {\bibinfo {title} {Riemannian optimization of photonic quantum circuits in phase and fock space},\ }\href@noop {} {\bibfield  {journal} {\bibinfo  {journal} {SciPost Physics}\ }\textbf {\bibinfo {volume} {17}},\ \bibinfo {pages} {082} (\bibinfo {year} {2024})}\BibitemShut {NoStop}%
\bibitem [{\citenamefont {Sun}\ \emph {et~al.}(2026)\citenamefont {Sun}, \citenamefont {Combes},\ and\ \citenamefont {Hackl}}]{sun2026representation}%
  \BibitemOpen
  \bibfield  {author} {\bibinfo {author} {\bibfnamefont {J.}~\bibnamefont {Sun}}, \bibinfo {author} {\bibfnamefont {J.}~\bibnamefont {Combes}},\ and\ \bibinfo {author} {\bibfnamefont {L.}~\bibnamefont {Hackl}},\ }\bibfield  {title} {\bibinfo {title} {Representation theory of inhomogeneous gaussian unitaries},\ }\href@noop {} {\bibfield  {journal} {\bibinfo  {journal} {arXiv preprint arXiv:2602.08611}\ } (\bibinfo {year} {2026})}\BibitemShut {NoStop}%
\bibitem [{\citenamefont {Bergmann}\ and\ \citenamefont {van Loock}(2016)}]{bergmann2016quantum}%
  \BibitemOpen
  \bibfield  {author} {\bibinfo {author} {\bibfnamefont {M.}~\bibnamefont {Bergmann}}\ and\ \bibinfo {author} {\bibfnamefont {P.}~\bibnamefont {van Loock}},\ }\bibfield  {title} {\bibinfo {title} {Quantum error correction against photon loss using multicomponent cat states},\ }\href@noop {} {\bibfield  {journal} {\bibinfo  {journal} {Physical Review A}\ }\textbf {\bibinfo {volume} {94}},\ \bibinfo {pages} {042332} (\bibinfo {year} {2016})}\BibitemShut {NoStop}%
\bibitem [{\citenamefont {Boon}\ \emph {et~al.}(2026)\citenamefont {Boon}, \citenamefont {Landon-Cardinal},\ and\ \citenamefont {Quesada}}]{boon2026generalised}%
  \BibitemOpen
  \bibfield  {author} {\bibinfo {author} {\bibfnamefont {A.~J.}\ \bibnamefont {Boon}}, \bibinfo {author} {\bibfnamefont {O.}~\bibnamefont {Landon-Cardinal}},\ and\ \bibinfo {author} {\bibfnamefont {N.}~\bibnamefont {Quesada}},\ }\bibfield  {title} {\bibinfo {title} {Generalised all-optical cat correction},\ }\href@noop {} {\bibfield  {journal} {\bibinfo  {journal} {arXiv preprint arXiv:2603.03263}\ } (\bibinfo {year} {2026})}\BibitemShut {NoStop}%
\bibitem [{\citenamefont {Reed}\ and\ \citenamefont {Simon}(1978)}]{reed1978iv}%
  \BibitemOpen
  \bibfield  {author} {\bibinfo {author} {\bibfnamefont {M.}~\bibnamefont {Reed}}\ and\ \bibinfo {author} {\bibfnamefont {B.}~\bibnamefont {Simon}},\ }\href@noop {} {\emph {\bibinfo {title} {IV: Analysis of Operators}}},\ Vol.~\bibinfo {volume} {4}\ (\bibinfo  {publisher} {Elsevier},\ \bibinfo {year} {1978})\BibitemShut {NoStop}%
\bibitem [{\citenamefont {Houde}\ \emph {et~al.}(2024)\citenamefont {Houde}, \citenamefont {McCutcheon},\ and\ \citenamefont {Quesada}}]{houde2024matrix}%
  \BibitemOpen
  \bibfield  {author} {\bibinfo {author} {\bibfnamefont {M.}~\bibnamefont {Houde}}, \bibinfo {author} {\bibfnamefont {W.}~\bibnamefont {McCutcheon}},\ and\ \bibinfo {author} {\bibfnamefont {N.}~\bibnamefont {Quesada}},\ }\bibfield  {title} {\bibinfo {title} {Matrix decompositions in quantum optics: Takagi/autonne, bloch--messiah/euler, iwasawa, and williamson},\ }\href@noop {} {\bibfield  {journal} {\bibinfo  {journal} {Canadian Journal of Physics}\ }\textbf {\bibinfo {volume} {102}},\ \bibinfo {pages} {497} (\bibinfo {year} {2024})}\BibitemShut {NoStop}%
\bibitem [{\citenamefont {Brask}(2021)}]{brask2021gaussian}%
  \BibitemOpen
  \bibfield  {author} {\bibinfo {author} {\bibfnamefont {J.~B.}\ \bibnamefont {Brask}},\ }\bibfield  {title} {\bibinfo {title} {Gaussian states and operations--a quick reference},\ }\href@noop {} {\bibfield  {journal} {\bibinfo  {journal} {arXiv preprint arXiv:2102.05748}\ } (\bibinfo {year} {2021})}\BibitemShut {NoStop}%
\end{thebibliography}%

\appendix
\newpage
    \section{Supplemental Material}
    In this Supplemental Material we show details on the derivation of the results of the main text. In particular, we show a proof for Theorem~\ref{theo:only_one_positive_eival}, a derivation of a variational lower bound on the trace distance between a pure and a mixed Gaussian state, and a proof for Eqs.~\eqref{eq:lanczos_vectors_update} to~\eqref{eq:tridiagonal_entries}. We also show a slight generalization of the problem of computing trace distances between linear combinations of Gaussian states.

    \subsection{1. Proof of Theorem~\ref{theo:only_one_positive_eival}}
    Let $\hat{\varrho}$ be a mixed state defined over the Hilbert space $\mathcal{H}=L^2(\mathbb{R}^M)$, and let $|\psi\rangle\in \mathcal{H}$ be a normalized vector representing a pure state. Since $\hat{\varrho}\geq0$, $\langle\psi|\hat{\varrho}|\psi\rangle\geq0$ for all $|\psi\rangle$. Moreover, using the Cauchy-Schwarz inequality, we have that $\langle\psi|\hat{\varrho}|\psi\rangle\leq \sqrt{\mathrm{Tr}(\hat{\varrho}\hat{\varrho}^\dagger)}\sqrt{\mathrm{Tr}(|\psi\rangle\langle \psi||\psi\rangle\langle \psi|)}= \sqrt{\mathrm{Tr}(\hat{\varrho}^2)}< 1$. As mentioned in the main text, the operator $|\psi\rangle\langle\psi|-\hat{\varrho}$, which is trace class (with trace $0$) and therefore compact, has a positive orthogonal decomposition of the form $|\psi\rangle\langle\psi|-\hat{\varrho} = \hat{P}-\hat{Q}$, where $\hat{P},\hat{Q}\geq 0$ and $\hat{P}\hat{Q}=\hat{Q}\hat{P} = 0$.

    Before beginningbegining the proof, we first state the \textit{Courant-Fisher theorem}, also known as the \textit{min-max principle}, whose proof can be found in Ref.~\cite{reed1978iv}.

    \begin{theorem}[Courant-Fischer]
        \label{theo:courant_fischer}
        Let $\hat{A}$ be a compact, Hermitian, linear operator acting on the Hilbert space $\mathcal{H}$. Arrange the corresponding \textit{positive} eigenvalues in non-increasing order as $\lambda_1\geq\lambda_2\geq\dots$, then
        \begin{equation}
            \lambda_k = \max_{S_k}\min_{\substack{|\psi\rangle\in S_k\\\||\psi\rangle\|=1}}\langle\psi|\hat{A}|\psi\rangle,
            \label{eq:CF_theorem_one}
        \end{equation}
        \begin{equation}
            \lambda_k = \min_{S_{k-1}}\max_{\substack{|\psi\rangle\in S_{k-1}^{\perp}\\\||\psi\rangle\|=1}}\langle\psi|\hat{A}|\psi\rangle,
            \label{eq:CF_theorem_two}
        \end{equation}
        where $S_k$ is any $k$-dimensional subspace of $\mathcal{H}$, and $S_k^{\perp}$ denotes the orthogonal complement of $S_k$. A similar set of equations can be obtained for the negative eigenvalues of $\hat{A}$ (by changing $\hat{A}\rightarrow-\hat{A}$).
    \end{theorem}

    Now we can prove Theorem~\ref{theo:only_one_positive_eival}, which, let us recall, states that the operator $|\psi\rangle\langle\psi|-\hat{\varrho}$ has only one positive eigenvalue.

    \begin{proof}
        Suppose that $|\psi\rangle\langle\psi|-\hat{\varrho}$ has more than one positive eigenvalue, so we can arrange them in non-increasing order $\lambda_1\geq\lambda_2\geq\dots$. Since $|\psi\rangle\langle\psi|-\hat{\varrho}$ is compact and Hermitian, we can use the Courant-Fischer theorem to compute all the possible eigenvalues. Using Eq.~\eqref{eq:CF_theorem_one}, we have that 
        \begin{equation}
            \lambda_1=\max_{S_1}\min_{\substack{|\chi\rangle\in S_1\\\||\chi\rangle\|=1}}\langle\chi|(|\psi\rangle\langle\psi|-\hat{\varrho})|\chi\rangle.
        \end{equation}
        Every one-dimensional subspace of $\mathcal{H}$ can be written as $\mathrm{span}(|\phi\rangle)$, for some arbitrary, not necessarily normalized, $|\phi\rangle\in \mathcal{H}$. Any normalized vector in $\mathrm{span}(|\phi\rangle)$ will have the form $|\chi\rangle=\alpha|\phi\rangle/\||\phi\rangle\|$, where $\alpha$ is a complex number satisfying $|\alpha|=1$. For all these vectors we have that $\langle\chi|\hat{A}|\chi\rangle=\langle\phi|\hat{A}|\phi\rangle/ \langle\phi|\phi\rangle$, for any $\hat{A}$. We may therefore write
        \begin{equation}
            \lambda_1=\max_{|\phi\rangle}\frac{\langle\phi|(|\psi\rangle\langle\psi|-\hat{\varrho})|\phi\rangle}{\langle\phi|\phi\rangle}.
        \end{equation}
        Let $|\phi\rangle=|\psi\rangle$, then we have that $\lambda_1\geq \langle\psi|(|\psi\rangle\langle\psi|-\hat{\varrho})|\psi\rangle = 1-\langle\psi|\hat{\varrho}|\psi\rangle >0$, since $0\leq\langle\psi|\hat{\varrho}|\psi\rangle<1$. This confirms that there is at least one positive eigenvalue of $|\psi\rangle\langle\psi|-\hat{\varrho}$.

        We move on now to compute $\lambda_2$. Using Eq.~\eqref{eq:CF_theorem_two}, we have that
        \begin{equation}
            \lambda_2 = \min_{S_{1}}\max_{\substack{|\chi\rangle\in S_{1}^{\perp}\\\||\chi\rangle\|=1}}\langle\chi|(|\psi\rangle\langle\psi|-\hat{\varrho})|\chi\rangle.
        \end{equation}
        Noting once more that every $S_1$ has the form $\mathrm{span}(|\phi\rangle)$, for some $|\phi\rangle\in \mathcal{H}$, we can write
        \begin{equation}
            \lambda_2 = \min_{|\phi\rangle}\max_{\substack{|\chi\rangle\perp |\phi\rangle\\\||\chi\rangle\|=1}}\langle\chi|(|\psi\rangle\langle\psi|-\hat{\varrho})|\chi\rangle.
        \end{equation}
        Again, let $|\phi\rangle=|\psi\rangle$. Since $|\chi\rangle\perp |\psi\rangle$, we have that $\langle\chi|\psi\rangle\langle\psi|\chi\rangle = 0$, and then 
        \begin{equation}
            \lambda_2 \leq \max_{\substack{|\chi\rangle\perp |\phi\rangle\\\||\chi\rangle\|=1}}-\langle\chi|\hat{\varrho}|\chi\rangle=-\min_{\substack{|\chi\rangle\perp |\phi\rangle\\\||\chi\rangle\|=1}}\langle\chi|\hat{\varrho}|\chi\rangle\leq 0.
        \end{equation}
        This is a contradiction, since we had assumed that $|\psi\rangle\langle\psi|-\hat{\varrho}$ had more than one \textit{strictly positive} eigenvalue. We therefore conclude that $|\psi\rangle\langle\psi|-\hat{\varrho}$ has only one positive eigenvalue, $\lambda_1=\lambda_+$, which implies that $\hat{P}$ is rank-1, i.e., $\hat{P}=\lambda_+|\lambda_+\rangle\langle\lambda_+|$.
    \end{proof}

    \subsection{2. A variational lower bound on the trace distance between a pure and a mixed Gaussian state}

    Consider now a pure Gaussian state $\ket{\psi}\bra{\psi} \leftrightarrow (\bm{V}_{\psi},  \bm{r}_{\psi})$ and a mixed state $\hat{\varrho} \leftrightarrow (\bm{V}_{\hat{\varrho}}, \bm{r}_{\hat{\varrho}})$, each of them characterized by their corresponding vectors of first moments and covariance matrices.
    
Using the invariance of the trace distance under unitary transformations, and the fact that any pure Gaussian state can be constructed by applying a unitary operation over the vacuum~\cite{serafini2023quantum}, it is not difficult to show that
\begin{align}
d(\hat{\varrho}, \ket{\psi} \bra{\psi}) = d(\hat{\tau}, \ket{0} \bra{0})
\end{align}
where $\hat{\tau} \leftrightarrow (\bm{V}, \bm{r})$ is a Gaussian state whose covariance matrix and vector of first moments can be, in at most cubic time in the number of modes, obtained from those of $\ket{\psi}\bra{\psi}$ and $\hat{\varrho}$. Moreover, $\ket{0}\bra{0}$ is the vacuum, a Gaussian state with $\bm{V}=\tfrac{\hbar}{2} \mathbb{I}_2, \bm{r}=\bm{0}$.

We would like to give a lower bound on the trace distance between the pure and the mixed state by proposing a variational ket. This ket must have a non-zero overlap with $\ket{0}$ and thus we propose that
\begin{align}\label{ansatz}
\ket{\nu} = \cos \tfrac{\theta}{2} \ket{0} + e^{i \varphi} \sin \tfrac{\theta}{2} \ket{\phi}
\end{align}
with $\braket{0|\phi}=0$. A reasonable proposal is to choose the second state, $\ket{\phi}$, to be a Gaussian state. However, no Gaussian state is orthogonal to the vacuum; this issue is easily fixed by choosing a Gaussian state $\ket{\mu}$ and projecting out its vacuum component to obtain 
\begin{align}
\ket{\phi} = \mathcal{N}\left( \ket{\mu} - \braket{0|\mu} \ket{0} \right),\,\, \mathcal{N} = \frac{1}{\sqrt{1-|\braket{0|\mu}|^2}}.
\end{align}
With the parametrization in Eq.~\eqref{ansatz} we obtain
\begin{align}
\braket{\nu|\left( \ket{0} \bra{0} - \hat{\tau}   \right) |\nu} =&  \tfrac{1}{2} \left( 1+ \cos \theta \right)- \tfrac{1}{2} \left( 1+ \cos \theta \right) \braket{0 | \hat{\tau} | 0 }\nonumber\\
& - \tfrac{1}{2} \left( 1- \cos \theta \right) \braket{\phi|\hat{\tau}|\phi} \nonumber\\
&- \sin \theta  \ \mathrm{Re}\left(e^{i \varphi} \braket{0|\hat{\tau}|\phi} \right).
\end{align}
For a fixed $\ket{\phi}$, extremizing the equation above corresponds to choosing $\varphi$ such that $e^{i \varphi} \braket{0|\hat{\tau}|\phi}$ is purely real and positive (and thus $| \braket{0|\hat{\tau}|\phi}|=e^{i \varphi} \braket{0|\hat{\tau}|\phi}$) and we can write
\begin{align}
&\max_{\varphi} \braket{\nu|\left(  \ket{0} \bra{0} - \hat{\tau} \right) |\nu} = \tfrac{1}{2}\left(1 - \braket{0|\hat{\tau}|0} - \braket{\phi|\hat{\tau}|\phi}  \right) \nonumber \\
&\quad -\underbrace{\tfrac{1}{2}(-1 + \braket{0|\hat{\tau}|0} - \braket{\phi|\hat{\tau}|\phi} )}_{\equiv x} \cos \theta  -\underbrace{ | \braket{0|\hat{\tau}|\phi}| }_{\equiv y}\sin \theta.
\end{align}
This quantity is maximized when $\cos \theta = -x/\sqrt{x^2+y^2}$ and $\sin \theta =-y/\sqrt{x^2+y^2}$ 
to obtain
\begin{align}
&\max_{\varphi, \theta} \braket{\nu|\left( \ket{0} \bra{0} - \hat{\tau}   \right) |\nu}=\tfrac{1}{2}\left(1 - \braket{0|\hat{\tau}|0} - \braket{\phi|\hat{\tau}|\phi}  \right) \nonumber \\
&+\tfrac{1}{2} \sqrt{4| \braket{0|\hat{\tau}|\phi}|^2 + (-1 + \braket{0|\hat{\tau}|0} - \braket{\phi|\hat{\tau}|\phi} )^2}.
\end{align}

We can now evaluate the necessary pieces for the maximization in terms of the two Gaussian states, $\ket{0}$ and $\ket{\mu}$
\begin{align}
\braket{0|\hat{\tau}|\phi} &= \mathcal{N}  \bra{0}\hat{\tau} \left[\ket{\mu} - \braket{0|\mu} \ket{0} \right] \nonumber\\
&= \mathcal{N} \left[ \braket{0|\hat{\tau}|\mu} - \braket{0|\mu} \braket{0|\hat{\tau}|0}\right],
\end{align}
\begin{align}
\braket{\phi|\hat{\tau}|\phi} &= \mathcal{N}^2 \left[ \braket{\mu|\hat{\tau}|\mu} - \braket{0|\mu} \braket{\mu|\hat{\tau}|0} \right.\nonumber\\
& \left. \quad \quad \quad - \braket{\mu|0} \braket{0|\hat{\tau}|\mu}+ |\braket{\mu|0}|^2 \braket{0|\hat{\tau}|0}\right].
\end{align}
The covariance matrix $\bm{V}$ of the mixed Gaussian state $\hat{\tau}$ can be, using Williamson's decomposition~\cite{serafini2023quantum,houde2024matrix}, written as
\begin{align}\label{eq:williamson}
\bm{V} = \tfrac{\hbar}{2}\bm{S} [\oplus_{i=1}^M (2 \bar{n}_i +1)] \oplus [\oplus_{i=1}^M (2 \bar{n}_i +1)] \bm{S}^T.
 \end{align}
If we select as Gaussian trial state $\ket{\mu}$ one with covariance matrix $\bm{V}_{\text{pure}} = \tfrac{\hbar}{2} \bm{S} \bm{S}^T$ and vector of means $\bm{r}$ (same as for the mixed state) then we obtain
\begin{align}
\braket{\mu|\hat{\tau}|\mu} = \braket{0|\tau_{\text{th}}|0} = F_{\text{th}} =  \frac{1}{\prod_{i=1}^M (1+\bar{n}_i)},
\end{align}
\begin{align}
\braket{\mu|\hat{\tau}|0} = \braket{\mu|0} \braket{0|\tau_{\text{th}}|0} = \sqrt{F_{\text{coh}}} \times   F_{\text{th}},
\end{align}
\begin{align}
\braket{0|\hat{\tau}|0} = F.
\end{align}
with (we use the square of the Fidelity used in Ref. ~\cite{brask2021gaussian} Eq. (117) and also we use an arbitrary $\hbar$ as opposed to their $\hbar=1$)
\begin{align}
F_{\text{coh}} =  \braket{0 |\mu} \braket{\mu|0 }  = \frac{\exp\left[-\tfrac12 \bar{\bm{r}}^T \left\{ \bm{V}_{\text{pure}} + \frac{\hbar}{2} \mathbb{I} \right\}^{-1} \bar{\bm{r}}\right]}{\sqrt{\det \left(\frac{1}{\hbar}\left[ \bm{V}_{\text{pure}} + \frac{\hbar}{2} \mathbb{I}\right]\right)}}.
\end{align}

Finally we arrive to
\begin{widetext}
\begin{align}
&\max_{\varphi, \theta} \braket{\nu|\left( \ket{0} \bra{0} -\hat{\tau}   \right) |\nu}=\frac{1}{2}-\frac{1}{2 (1-F_{\text{coh}})}  \Bigg[ F+F_{\text{th}}-2 F_{\text{coh}} F_{\text{th}} \Bigg. \nonumber \\
&- \sqrt{F_{\text{coh}}^2 \left(-4 F+4 F_{\text{th}}+1\right)+F_{\text{coh}} \left(6 F-6 F_{\text{th}}-2\right)+\left(-F+F_{\text{th}}+1\right){}^2} \Bigg. \Bigg].
\end{align}
\end{widetext}
For pure states $F =F_{\text{coh}}$ and $F_{\text{th}}=1$ and the quantity above becomes $\sqrt{1-F_{\text{coh}}}$ which is precisely the value of the trace distance for two pure states.

\subsection{3. Lanczos algorithm relations}
Let us set $\hat{A}=|\psi\rangle\langle\psi|-\hat{\varrho}$ and $|c\rangle=|\psi\rangle$ in Lanczos algorithm. Then, the following proposition holds:
\begin{prop}
    $(|\psi\rangle\langle\psi|-\hat{\varrho})^\ell|\psi\rangle = \sum_{k=0}^\ell c_{\ell, k}\hat{\varrho}^k|\psi\rangle$, where $c_{\ell,k}\in\mathbb{R}$ for all $k$.
\end{prop}

\begin{proof}
    We proceed by induction. For $\ell=1$, we have that $(|\psi\rangle\langle\psi|-\hat{\varrho})^\ell|\psi\rangle = |\psi\rangle-\hat{\varrho}|\psi\rangle$, so $c_{1,0}=1$, $c_{1,1}=-1$. For $\ell=2$, we find $(|\psi\rangle\langle\psi|-\hat{\varrho})^\ell|\psi\rangle = (1-\langle\psi|\hat{\varrho}|\psi\rangle)|\psi\rangle-\hat{\varrho}|\psi\rangle+\hat{\varrho}^2|\psi\rangle$, which implies that $c_{2,0}=1-\langle\psi|\hat{\varrho}|\psi\rangle$, $c_{2,1}=-1$, $c_{2,2}=1$. Suppose that this relation holds for an arbitrary $\ell$, then for $\ell+1$, we have
    \begin{align}
        (|\psi\rangle&\langle\psi|-\hat{\varrho})^{\ell+1}|\psi\rangle = (|\psi\rangle\langle\psi|-\hat{\varrho})\sum_{k=0}^\ell c_{\ell,k}\hat{\varrho}^k|\psi\rangle\nonumber\\
        &=\left(\sum_{k=0}^\ell c_{\ell,k}\langle\psi|\hat{\varrho}^k|\psi\rangle\right)|\psi\rangle - \sum_{k=0}^\ell c_{\ell,k}\hat{\varrho}^{k+1}|\psi\rangle\nonumber\\
        &=c_{l+1,0}|\psi\rangle - \sum_{j=1}^{\ell+1} c_{\ell,j-1}\hat{\varrho}^{j}|\psi\rangle\nonumber\\
        &=c_{l+1,0}|\psi\rangle + \sum_{j=1}^{\ell+1} c_{\ell+1,j}\hat{\varrho}^{j}|\psi\rangle = \sum_{j=0}^{l+1}c_{l+1,j}\hat{\varrho}^j|\psi\rangle,
    \end{align}
    which completes the proof. From this procedure, we obtain the recurrence relations $c_{\ell,0}=\sum_{k=0}^{\ell-1} c_{\ell-1,k}\langle\psi|\hat{\varrho}^k|\psi\rangle$, $c_{\ell,j}=-c_{\ell-1,j-1}$ for $j\geq 1$, with the initial condition $c_{0,0}=1$.
  
\end{proof}

    Now, the vectors $|\tilde{\phi}_k\rangle$ are defined through the relation
    \[|\tilde{\phi}_k\rangle = (|\psi\rangle\langle\psi|-\hat{\varrho})^k|\psi\rangle - \sum_{j=0}^{k-1}\langle\phi_j|(|\psi\rangle\langle\psi|-\hat{\varrho})^k|\psi\rangle|\phi_j\rangle.\]
    The following result holds.
    \begin{prop}
        We have that $|\tilde{\phi}_k\rangle = \sum_{j=0}^k \tilde{d}_{k, j}\hat{\varrho}^j|\psi\rangle$, where $\tilde{d}_{k,j}\in\mathbb{R}$ for all $j$. After normalization, it holds that $|\phi_k\rangle = \sum_{j=0}^k d_{k, j}\hat{\varrho}^j|\psi\rangle$
    \end{prop}
    \begin{proof}
        From the definition, and assuming that the relation holds up to $k-1$, we can see that
        \begin{align}
            &|\tilde{\phi}_k\rangle = \sum_{b=0}^kc_{k,b}\hat{\varrho}^b|\psi\rangle \nonumber\\
            &-\sum_{j=0}^{k-1}\left(\sum_{a=0}^jd_{j,a}\langle\psi|\hat{\varrho}^a\right)\left(\sum_{b=0}^kc_{k,b}\hat{\varrho}^b|\psi\rangle\right)\left(\sum_{r=0}^jd_{j,r}\hat{\varrho}^r|\psi\rangle\right)\nonumber\\
            &=\sum_{b=0}^kc_{k,b}\hat{\varrho}^b|\psi\rangle\nonumber\\
            &-\sum_{j=0}^{k-1}\left(\sum_{a=0}^j\sum_{b=0}^kd_{j,a}\langle\psi|\hat{\varrho}^{a+b}|\psi\rangle c_{k,b}\right)\sum_{r=0}^jd_{j,r}\hat{\varrho}^r|\psi\rangle.
        \end{align}
        Define the matrices $\bm{\mathcal{G}}$, $\bm{D}$, $\bm{C}$, with components $(\bm{\mathcal{G}})_{a,b}=\langle\psi|\hat{\varrho}^{a+b}|\psi\rangle$, $D_{j,k}=d_{j,k}$ and $C_{j,k}=c_{j,k}$. Since $D_{j,k}=0$ for $k>j$, we have that
        \begin{align}
            \sum_{a=0}^j\sum_{b=0}^kd_{j,a}\langle\psi|\hat{\varrho}^{a+b}|\psi\rangle c_{k,b}&=\sum_{a=0}^k\sum_{b=0}^kd_{j,a}\langle\psi|\hat{\varrho}^{a+b}|\psi\rangle c_{k,b}\nonumber\\
            &=(\bm{D}\bm{\mathcal{G}}\bm{C}^\mathrm{T})_{j,k}.
        \end{align}
        Then, 
        \begin{equation}
            |\tilde{\phi}_k\rangle = \sum_{b=0}^kC_{k,b}\hat{\varrho}^b|\psi\rangle -\sum_{j=0}^{k-1}\sum_{r=0}^j(\bm{D}\bm{\mathcal{G}}\bm{C}^\mathrm{T})_{j,k}D_{j,r}\hat{\varrho}^r|\psi\rangle.
        \end{equation}
        Interchanging the sums over $j$ and $r$, we find
        \begin{align}
            \sum_{j=0}^{k-1}\sum_{r=0}^j&(\bm{D}\bm{\mathcal{G}}\bm{C}^\mathrm{T})_{j,k}D_{j,r}\hat{\varrho}^r|\psi\rangle\nonumber\\
            &=\sum_{r=0}^{k-1}\left(\sum_{j=r}^{k-1}(\bm{D}^\mathrm{T})_{r,j}(\bm{D}\bm{\mathcal{G}}\bm{C}^\mathrm{T})_{j,k}\right)\hat{\varrho}^r|\psi\rangle.
        \end{align}
        Since $(\bm{D}^\mathrm{T})_{r,j}=0$ for $r<j$, we can write
        \begin{align}
            \sum_{j=r}^{k-1}(\bm{D}^\mathrm{T})_{r,j}(\bm{D}\bm{\mathcal{G}}\bm{C}^\mathrm{T})_{j,k}&=\sum_{j=0}^{k-1}(\bm{D}^\mathrm{T})_{r,j}(\bm{D}\bm{\mathcal{G}}\bm{C}^\mathrm{T})_{j,k}\nonumber\\
            &=(\bm{D}^\mathrm{T}\bm{D}\bm{\mathcal{G}}\bm{C}^\mathrm{T})_{r,k}
        \end{align}
        because, up to step $k-1$, $D_{j,k}=0$ for all $j\geq k$. This leads to the relation
        \begin{align}
            |\tilde{\phi}_k\rangle &= \sum_{b=0}^kC_{k,b}\hat{\varrho}^b|\psi\rangle -\sum_{r=0}^{k-1}(\bm{D}^\mathrm{T}\bm{D}\bm{\mathcal{G}}\bm{C}^\mathrm{T})_{r,k}\hat{\varrho}^r|\psi\rangle\nonumber\\
            &=\sum_{b=0}^kC_{k,b}\hat{\varrho}^b|\psi\rangle -\sum_{r=0}^{k-1}(\bm{C}\bm{\mathcal{G}}\bm{D}^\mathrm{T}\bm{D})_{k,r}\hat{\varrho}^r|\psi\rangle.
        \end{align}
        Finally, we can see that
        \begin{align}
            |\tilde{\phi}_k\rangle &= \sum_{b=0}^{k-1}(\bm{C}-\bm{C}\bm{\mathcal{G}}\bm{D}^\mathrm{T}\bm{D})_{k,b}\hat{\varrho}^b|\psi\rangle + C_{k,k}\hat{\varrho}^k|\psi\rangle\nonumber\\
            &=\sum_{b=0}^k\tilde{d}_{k,b}\hat{\varrho}^b|\psi\rangle,
        \end{align}
        where we find the recurrence relations $\tilde{D}_{k,k}=\tilde{d}_{k,k}=C_{k,k}$, $\tilde{D}_{k,b}=(\bm{C}-\bm{C}\bm{\mathcal{G}}\bm{D}^\mathrm{T}\bm{D})_{k,b}$ for $b\in\{0,\dots,k-1\}$, with the initial condition $\tilde{D}_{0,0}=1$. After normalization, we find the expression for $|\phi_k\rangle$.
    \end{proof}

    \subsection{4. A generalization of the algorithm for computing trace distances of linear combinations of Gaussian states}

    Suppose that $|\psi\rangle\langle\psi|$ is a linear combination of pure Gaussian states, while $\hat{\varrho}$ can be written as a linear combination of \textit{products} of Gaussian states. Namely, suppose that $|\psi\rangle = \sum_{j=1}^p a_j|g_j\rangle$ with $a_j\in \mathbb{C}$ and $|g_j\rangle$ a pure Gaussian state parametrized by the pair $(\bm{r}_j, \bm{V}_j)$. As before, we assume that the $\{a_j\}$ are chosen so that $|\psi\rangle$ is normalized and, without loss of generality, that $\langle0|g_j\rangle\geq 0$. On the other hand, we write $\hat{\varrho} = \sum_{k=1}^q b_k\hat{\nu}_k$, where $b_k\in\mathbb{C}$ and $\hat{\nu}_k$ is a \textit{product of Gaussian states}. Here, we also assume that the $\{b_k\}$ are chosen so that $\hat{\varrho}$ is a valid quantum state. By being a product of Gaussian states, each $\hat{\nu}_k$ is parametrized by a set of ordered pairs $\{(\bar{\bm{r}}^{(k)}_n, \bar{\bm{V}}^{(k)}_n)\}$, were the index $n$ runs from 1 up to the number of states composing $\hat{\nu}_k$. 
    
    A good example of this type of state is the two-component cat state $|\psi\rangle\langle\psi|\propto |\alpha\rangle\langle\alpha| + |-\alpha\rangle\langle-\alpha| + |\alpha\rangle\langle-\alpha| + |-\alpha\rangle\langle\alpha|$, which can be turned into the form described above by noticing that $|\alpha\rangle\langle-\alpha| = |\alpha\rangle\langle\alpha|\cdot|-\alpha\rangle\langle-\alpha| / \langle\alpha|-\alpha\rangle = e^{2|\alpha|^2}|\alpha\rangle\langle\alpha|\cdot|-\alpha\rangle\langle-\alpha|$. We explicitly introduce the product symbol between projectors to indicate that we are considering the product of two different operators instead of the outer product between two vectors.
    
    Using these expressions for $|\psi\rangle$ and $\hat{\varrho}$, one may show that
    \begin{align}
        \langle\psi|\hat{\varrho}^\ell|\psi\rangle = \sum_{\bm{u}\in [q]^\ell}&\prod_{i=1}^\ell b_{u_i}\sum_{j,k=1}^p \frac{a_j^*a_k}{\langle g_k|g_j\rangle}\nonumber\\
        &\times\mathrm{Tr}(\hat{\nu}_{u_1}\cdots\hat{\nu}_{u_\ell}|g_k\rangle\langle g_k|\cdot|g_j\rangle\langle g_j|),
        \label{eq:linear_comb_Bargmann_invariants}
    \end{align}
    where $[q]=\{1,\dots,q\}$, so $\bm{u}$ is a tuple of length $\ell$ whose entries are in $[q]$. Since the term inside the trace is always a product of Gaussian states, the computation of the metric boils down to a summation of several Gaussian Bargmann invariants that follow Eq.~\eqref{eq:bargmann_invariant}. However, the number of elements in this sum grows as $\mathcal{O}(q^\ell)$, making the time of computation of the metric exponential in the number of steps the algorithm runs for.

\end{document}